\documentclass[%
reprint,
nofootinbib,
amsmath,amssymb,
aps,
amsthm,
pra,
]{revtex4-1}

\usepackage{ifthen}
\newboolean{arxiv}
\setboolean{arxiv}{true}
\ifthenelse{\boolean{arxiv}}{\pdfoutput=1}{}

\usepackage{theorem} % \theorembodyfont
\usepackage{ascmac}
\usepackage{bbm}
\usepackage{bm} % bold math
\usepackage{braket}
\usepackage{dcolumn} % Align table columns on decimal point
\usepackage{enumerate}
\usepackage{enumitem}
\usepackage{framed}
\usepackage{graphicx} % Include figure files
\ifthenelse{\boolean{arxiv}}{\usepackage{hyperref}}{\usepackage[dvipdfmx]{hyperref}}
\usepackage{longtable}
\usepackage{mathtools}
\usepackage{multirow}
\ifthenelse{\boolean{arxiv}}{}{\usepackage{pxjahyper}}
\usepackage{txfonts} % \coloneqq,\llbracket
\usepackage{xparse}

\usepackage{color}
\usepackage[most]{tcolorbox}

\ifthenelse{\boolean{arxiv}}{}{\mathtoolsset{showonlyrefs}}

\newcommand{\hypercolor}{blue}
\hypersetup{
  colorlinks,
  citecolor=\hypercolor,
  linkcolor=\hypercolor,
  urlcolor=\hypercolor,
  bookmarksopen,
  bookmarksopenlevel=4,
  bookmarksnumbered
}

\newboolean{figure}
\setboolean{figure}{true}

\usepackage{bm}

\theorembodyfont{\upshape}
\newtheorem{thm}{Theorem}

\newtheorem{postulateno}{Postulate}

\newtheorem{lemma}[thm]{Lemma}
\newtheorem{cor}[thm]{Corollary}
\newcounter{proof}

\ExplSyntaxOn
\NewDocumentEnvironment{proof}{o}
 {
  \par\medskip
  \noindent
  \textbf{Proof \IfNoValueTF{#1}{}{~of~#1} ~}
 }
 {\QED\par\smallskip}
\ExplSyntaxOff

\newcounter{postulate}
\renewcommand{\thepostulate}{\arabic{postulate}}
\ExplSyntaxOn
\NewDocumentEnvironment{postulate}{oo}
 {
  \refstepcounter{postulate}
  \begin{postulateno}
  \textbf{\hspace{-0.5em}\IfNoValueTF{#2}{\thepostulate}{#2} ~(\IfNoValueTF{#1}{}{#1})}
 }
 {
  \end{postulateno}
 }
\ExplSyntaxOff

\newcommand{\mD}{\mathcal{D}}
\newcommand{\mE}{\mathcal{E}}

\newcommand{\Natural}{\mathbb{N}}

\newcommand{\QED}{\hspace*{0pt}\hfill $\blacksquare$}

\DeclareMathOperator*{\argmax}{argmax}

\DeclarePairedDelimiter{\gauss_sym}{\lfloor}{\rfloor}
\DeclareMathAlphabet{\mymathbb}{U}{BOONDOX-ds}{m}{n}
\newcommand{\zero}{\mymathbb{0}}

\newcommand{\titlename}{Unambiguous discrimination of the change point for quantum channels}

\newcommand{\cI}{\mathcal{I}}
\newcommand{\cP}{\mathcal{P}}

\newcommand{\I}{I}
\newcommand{\V}{V}
\newcommand{\W}{W}
\newcommand{\bV}{\mathbf{V}}
\newcommand{\bW}{\mathbf{W}}
\newcommand{\Pos}{\mathsf{Pos}}
\newcommand{\ot}{\otimes}
\newcommand{\tpsi}{\tilde{\psi}}
\newcommand{\tvarphi}{\tilde{\varphi}}
\newcommand{\tG}{\tilde{G}}
\newcommand{\tp}{\tilde{p}}

\newcommand{\tD}{\tilde{D}}
\newcommand{\tU}{\tilde{U}}
\newcommand{\w}{\omega}
\renewcommand{\ol}[1]{\overline{#1}}
\newcommand{\op}{\ol{p}}
\newcommand{\opt}{\star}
\newcommand{\oP}{\ol{P}}
\newcommand{\uP}[1]{\ol{P^{(#1)}}}
\newcommand{\lP}[1]{\underline{P^{(#1)}}}
\newcommand{\question}{\text{``?''}}

\newcommand{\relmiddle}[1]{\mathrel{}\middle#1\mathrel{}}
\newcommand{\relmid}{\relmiddle{\vert}}
\let\ast\relax
\DeclareMathOperator{\ast}{\circledast}
\setlist[enumerate]{label=\arabic*), leftmargin=3em, itemsep=0pt, parsep=0pt, labelwidth=5em}

\definecolor{memo}{RGB}{128,0,255}
\definecolor{gray}{RGB}{128,128,128}

\newcommand{\Discard}[1]{}

\newcommand{\EN}[1]{}

\begin{document}

\preprint{APS/123-QED}

\title{\titlename}

\affiliation{%
 Quantum Information Science Research Center, Quantum ICT Research Institute, Tamagawa University,
 Machida, Tokyo 194-8610, Japan
}%

\author{Kenji Nakahira}
\affiliation{%
 Quantum Information Science Research Center, Quantum ICT Research Institute, Tamagawa University,
 Machida, Tokyo 194-8610, Japan
}%

\date{\today}% It is always \today, today,

\begin{abstract}
 Identifying the precise moment when a quantum channel undergoes a change is
 a fundamental problem in quantum information theory. 
 We study how accurately one can determine the time at which a channel transitions to another.
 We investigate the quantum limit of the average success probability in unambiguous discrimination,
 in which errors are completely avoided by allowing inconclusive results
 with a certain probability.
 This problem can be viewed as a quantum process discrimination task,
 where the process consists of a sequence of quantum channels; however,
 obtaining analytical solutions for quantum process discrimination is generally extremely challenging.
 In this paper, we propose a method to derive lower and upper bounds on
 the maximum average success probability in unambiguous discrimination.
 In particular, when the channels before and after the change are unitary,
 we show that the maximum average success probability can be analytically expressed
 in terms of the length of the channel sequence and the discrimination limits for the two channels.
\end{abstract}

% PACS 03.67.Hk: Quantum communication
\maketitle
%\bookmarksetup{startatroot}

%\onecolumngrid

%\setcounter{tocdepth}{1}
%\tableofcontents

\section{Introduction}

The time at which the properties of an object change can carry important information.
If such time can be detected with high accuracy, then it may lead to new discoveries
or increase the likelihood of responding appropriately to the change.
The problem of detecting the time at which the properties of an object change abruptly
is known as the change point problem, and it has been studied in a wide range of fields,
including statistical analysis \cite{Page-1954,Page-1955,Bro-Dar-1993,Car-Mul-Sie-1994}.
Considering situations such as detecting changes in microscopic objects or weak signals,
it is natural to explore a quantum version of this problem.

A quantum version of the change point problem can be formulated as follows.
Suppose a quantum process composed of $N$ operations is applied sequentially.
For the first $k$ evaluations, the operation is implemented using a channel $\Lambda_0$,
then switches to a channel $\Lambda_1$ immediately after the $k$-th evaluation,
and remains $\Lambda_1$ for all subsequent evaluations.
For the first $k$ evaluations, the operation is realized by a channel $\Lambda_0$,
then transitions to another channel $\Lambda_1$ after the $k$-th evaluation,
and remains governed by the channel $\Lambda_1$ thereafter.
The goal is to identify the change point $k \in \{0,1,\dots,N\}$ as accurately as possible.
For example, when $N = 3$, the problem becomes one of discriminating
between the channel sequences $(\Lambda_1,\Lambda_1,\Lambda_1)$,
~$(\Lambda_0,\Lambda_1,\Lambda_1)$, ~$(\Lambda_0,\Lambda_0,\Lambda_1)$,
and $(\Lambda_0,\Lambda_0,\Lambda_0)$.
It is assumed that $\Lambda_0$ and $\Lambda_1$ are known, and that the change point is
equally likely to occur at any evaluation step.
If $\Lambda_0$ and $\Lambda_1$ can be perfectly distinguished in a single shot,
then the change point $k$ can be identified with certainty.
However, if this is not the case, then the optimal discrimination must be performed
based on a certain evaluation criterion.
This change point identification problem can be regarded as a special case of
quantum process discrimination.
Moreover, quantum states can be considered as a special case of quantum channels.

The simplest strategy in discrimination problems is arguably the one
that maximizes the average success probability, known as minimum-error discrimination
\cite{Hel-1976,Hol-1973,Yue-Ken-Lax-1975}.
The quantum change point problem has been studied under this minimum-error strategy
\cite{Sen-Bag-Cal-Chi-2016,Nak-2022-cp}.
Another representative strategy in discrimination problems is the unambiguous strategy,
which has also been widely investigated \cite{Iva-1987,Die-1988,Per-1988}.
In this strategy, inconclusive results, i.e., answers of ``unknown'', may be returned,
but whenever a conclusive result is given, it is guaranteed to be correct.
In other words, discrimination without error is achieved.
One potential application of the unambiguous strategy is efficient eavesdropping
in cryptographic systems \cite{Eke-Hut-Pal-Per-1994,Dus-Jah-Lut-2000}.
For the quantum change point problem, the unambiguous strategy has been studied
in the case in which both $\Lambda_0$ and $\Lambda_1$ are pure states.
In that case, the problem has been formulated as a semidefinite programming problem,
and an analytical solution for the maximum average success probability has been obtained
by considering its dual problem \cite{Sen-Cal-Mun-2017}.
However, to the best of our knowledge, the maximum average success probability for mixed states
or more general channels has not been known.

Here, we consider the unambiguous strategy in the case in which
both $\Lambda_0$ and $\Lambda_1$ are general quantum channels.
This discrimination can be modeled using the concept of a tester
\cite{Chi-Dar-Per-2008,Chi-Dar-Per-2009}, which is also referred to as a quantum strategy
\cite{Gut-Wat-2007} (see Fig.~\ref{fig:intro}).
When both $\Lambda_0$ and $\Lambda_1$ are quantum states, the problem reduces to
quantum state discrimination, and the tester corresponds simply to a quantum measurement.
In contrast, when they are general quantum channels, the tester consists of
multiple quantum operations, including channels and measurements,
which significantly complicates its optimization.
For instance, one must consider all possible discrimination strategies,
including those that use entangled input states with ancillary systems and
adaptive strategies.
Although this problem can be formulated as a semidefinite programming problem
\cite{Nak-Kat-2021-general}, obtaining analytical solutions is generally extremely challenging.
Therefore, we explore an alternative approach that differs from semidefinite programming.
In this paper, we present a method for deriving upper and lower bounds on
the maximum average success probability that applies when $\Lambda_0$ and $\Lambda_1$ are
arbitrary channels, including those involving quantum states.
Any tester for the change point problem, when used to discriminate between
the two channels $\Lambda_0$ and $\Lambda_1$, can never outperform
the optimal discrimination performance for these channels.
The upper bound is derived by considering only this natural constraint.
Importantly, although the upper bound is derived from a simple and intuitive constraint,
we show that the proposed upper bound coincides with the maximum average success probability
when both $\Lambda_0$ and $\Lambda_1$ are unitary channels.
This leads to an analytical expression for the maximum average success probability
in this case.
In process discrimination problems beyond state discrimination, analytical solutions
for the maximum average success probability have been obtained only in very special cases,
such as binary scenarios or highly symmetric settings (e.g., \cite{Zim-Sed-2010,Chi-Dar-Roe-2013}).
The present result constitutes a rare case in which an analytical solution is obtained
for multiple process discrimination problems that lack such symmetry.
\begin{figure}[bt]
 \centering
 \includegraphics[scale=1.0]{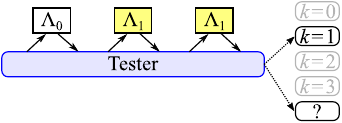}
 \caption{Identification of the change point with zero error
 (in the case of correct answer $k = 1$ for $N = 3$).
 Such discrimination can be represented using a tester.
 In this case, the unambiguous strategy returns either the result $1$ or an inconclusive result.}
 \label{fig:intro}
\end{figure}

\section{Upper bound of the maximum average success probability} \label{sec:UB}

For each $k \in \{0,\dots,N\}$, consider a process
$\mE_k \coloneqq (\Lambda_{k < 1}, \Lambda_{k < 2}, \dots, \Lambda_{k < N})$
consisting of $N$ channels.
The $n$-th channel of $\mE_k$ is $\Lambda_{k < n} \in \{ \Lambda_0, \Lambda_1 \}$,
where $\Lambda_{k < n}$ is $\Lambda_1$ if $k < n$ holds, and $\Lambda_0$ otherwise.
The change point problem can be interpreted as the task of discriminating
between the $N+1$ processes $\{ \mE_k \}_{k=0}^N$.
Let $P^{(N)}$ denote the maximum average success probability for this problem.

To derive an upper bound for $P^{(N)}$, consider the problem of discriminating between
two channels $\Lambda_0$ and $\Lambda_1$, and let $f(p)$ denote the maximum probability
of correctly identifying $\Lambda_1$ under the constraint that
$\Lambda_0$ is correctly identified with probability $p$.
Let $\op$ denote the maximum allowable value of $p$.
In the case of unambiguous discrimination, $\op$ is generally less than $1$.
If $f(0) > \op$ holds, then we swap $\Lambda_0$ and $\Lambda_1$ so that $f(0) \le \op$ holds.
For any $i \in \{1,\dots,N\}$, the channels $\Lambda_0$ and $\Lambda_1$ can be
transformed into the processes $\mE_{i-1}$ and $\mE_i$, respectively,
by inserting $i-1$ instances of $\Lambda_0$ before the channel and
$N-i$ instances of $\Lambda_1$ after it.
This transformation allows the discrimination between $\Lambda_0$ and $\Lambda_1$
to be reinterpreted as that between $\mE_{i-1}$ and $\mE_i$.
Therefore, for any fixed discrimination strategy in the change point problem,
when the probabilities of correctly identifying each process $\mE_0,\dots,\mE_N$ are
denoted by $p_0,\dots,p_N$, we have
\begin{alignat}{1}
 p_{i-1} \le f(p_i), \quad \forall i \in \{1,\dots,N\}.
 \label{eq:p_constraint}
\end{alignat}
Indeed, when $\Lambda_0$ and $\Lambda_1$ are replaced by $\mE_{i-1}$ and $\mE_i$,
respectively, a correct identification of $\mE_{i-1}$ implies
a correct identification of $\Lambda_1$, and a correct identification of $\mE_i$
implies a correct identification of $\Lambda_0$.
Since the success probability of identifying $\Lambda_0$ is $p_i$, the success probability
$p_{i-1}$ of identifying $\Lambda_1$ must not exceed $f(p_i)$.

Now, consider the problem of maximizing the average $\sum_{i=0}^N p_i / (N+1)$
under the constraint that $p_0,\dots,p_N$ satisfy Eq.~\eqref{eq:p_constraint}.
This optimization problem is formulated as follows:
\begin{alignat}{1}
 \begin{array}{ll}
  \text{maximize} & \displaystyle \frac{1}{N+1} \sum_{i=0}^N p_i \\
  \text{subject~to} & \displaystyle p_N \in [0,\op], ~ p_{i-1} \in [0,f(p_i)]
   ~(\forall i \in \{1,\dots,N\}), \\
 \end{array}
 \label{prob:p0}
\end{alignat}
where $[a,b]$ denotes the set of all real numbers between $a$ and $b$, inclusive.
Let $\uP{N}$ denote the optimal value of this problem.
Any discrimination strategy in the change point problem yields a success probability
vector $(p_0,\dots,p_N)$ that satisfies the constraints of this optimization problem,
and thus $P^{(N)} \le \uP{N}$ holds.
Let $f^{(i)}(p)$ be the result of applying $f$ iteratively $i$ times to $p$,
with $f^{(0)}(p) = p$.
Then, the following theorem holds (see Appendix~\ref{sec:R} for the proof):
\begin{thm} \label{thm:R}
 The following equation holds:
 \begin{alignat}{1}
  \uP{N} &= \max_{p \in [0,\op]} \frac{1}{N+1} \sum_{i=0}^N f^{(i)}(p).
  \label{eq:uPN}
 \end{alignat}
\end{thm}

This theorem implies that, once the map $f \colon [0,\op] \to [0,\op]$ is determined,
$\uP{N}$ is readily computable, at least from a computational standpoint.
Moreover, the theorem is highly versatile.
For example, if $\Lambda_0$ and $\Lambda_1$ represent states, then the result also applies to
the change point problem for states.
Furthermore, the inequality $P^{(N)} \le \uP{N}$ and this theorem hold not only for
unambiguous strategies but also for minimum-error strategies.

In particular, consider the case in which the map $f$ is an involution,
meaning that $f$ is invertible and its inverse is $f$ itself.
This corresponds to the situation in which the channels $\Lambda_0$ and $\Lambda_1$ are
symmetric.
In this case, from $f^{(i+2)} = f^{(i)}$, we have
\begin{alignat}{1}
 \uP{N} &= \oP_{\xi_N} \coloneqq \max_{p \in [0,\op]} [\xi_N p + (1-\xi_N) f(p)], \nonumber \\
 \xi_N &\coloneqq \frac{\gauss_sym{N/2}+1}{N+1},
 \label{eq:uPN_sym}
\end{alignat}
where $\gauss_sym{x}$ denotes the greatest integer less than or equal to $x$.
Therefore, $\uP{N}$ is equal to the maximum average success probability of discriminating
between the two channels $\Lambda_0$ and $\Lambda_1$ with prior probabilities
$\xi_N$ and $1-\xi_N$, respectively.
Now, consider the ordering of $\uP{1},\uP{2},\uP{3},\dots$.
It can be shown that $\oP_q$ is monotonically nondecreasing for $q \ge 1/2$%
\footnote{Proof: For $q \ge 1/2$, there exists $p \in [0,\op]$ satisfying
$\oP_q = q p + (1-q) f(p)$ and $p \ge f(p)$.
Therefore, for any positive real number $\delta$, we have
$\oP_{q+\delta} - \oP_q \ge \delta [p - f(p)] \ge 0$.}.
Moreover, $\xi_N = 1/2$ holds if $N$ is odd, and $\xi_N = (N+2)/(2N+2) > 1/2$ holds
otherwise.
Thus, we obtain
\begin{alignat}{1}
 \uP{1} &= \uP{3} = \cdots = \oP_{1/2}, \nonumber \\
 \uP{2} &\ge \uP{4} \ge \cdots \ge \oP_{1/2}.
 \label{eq:uP_oscillate}
\end{alignat}

\section{Unambiguous strategies for unitary channel sequences} \label{sec:unitary}

We here consider unambiguous discrimination in the specific case in which
$\Lambda_0$ and $\Lambda_1$ are distinct unitary channels.
Let $U_0$ and $U_1$ denote the unitary matrices representing them, respectively.
Then, since $f$ is an involution, Eq.~\eqref{eq:uPN_sym} holds,
which implies that $\uP{N}$ is equal to $\oP_{\xi_N}$, i.e., the maximum average success probability of
discriminating between the two channels $\Lambda_0$ and $\Lambda_1$
with prior probabilities $\xi_N$ and $1 - \xi_N$, respectively.
Let $t$ denote the minimum value of $|\braket{\varphi|U_0^\dagger U_1|\varphi}|$
with respect to all possible choices of the state $\ket{\varphi}$,
where $X^\dagger$ denotes the conjugate transpose of $X$.
$t$ and $1 - t$ are, respectively, equal to the minimum inconclusive probability
and the maximum average success probability when discriminating between the two channels
$\Lambda_0$ and $\Lambda_1$ with equal prior probabilities
using an unambiguous strategy \cite{Zim-Sed-2010}.
Let
\begin{alignat}{1}
 c &\coloneqq
 \begin{dcases}
  1, & N ~\text{is~odd}, \\
  \max \left\{ t,\sqrt{N/(N+2)} \right\}, & N ~\text{is~even};
 \end{dcases}
 \label{eq:c}
\end{alignat}
then, $p^\opt \coloneqq 1 - tc$ is the optimal solution to the problem in
Eq.~\eqref{eq:uPN_sym} and $f(p^\opt) = 1 - tc^{-1}$ holds \cite{Zim-Sed-2010}.
Substituting these into Eq.~\eqref{eq:uPN_sym}, we obtain
\begin{alignat}{1}
 \uP{N} &= \xi_N (1 - tc) + (1 - \xi_N) (1 - tc^{-1}) \nonumber \\
 &= 1 - t \cdot \frac{(N+2)c + Nc^{-1}}{2(N+1)}.
 \label{eq:uP_uni_unamb}
\end{alignat}
Let $\Gamma$ denote the polygon on the complex plane spanned by the eigenvalues of $U_0^\dagger U_1$.
Then, $t$ is equal to the distance from the origin to $\Gamma$ \cite{Dar-Pre-Par-2001}.
Specifically, let $\lambda_0$ and $\lambda_1$ be the eigenvalues corresponding to
the endpoints of the edge (i.e., line segment between two adjacent eigenvalues)
of $\Gamma$ that is closest to the origin,
and let $\ket{\lambda_0}$ and $\ket{\lambda_1}$ be the corresponding eigenvectors
of $U_0^\dagger U_1$.
Then, the minimum value of $|\braket{\varphi|U_0^\dagger U_1|\varphi}|$ is attained
when $\ket{\varphi}$ is equal to
\begin{alignat}{1}
 \ket{+} \coloneqq \frac{\ket{\lambda_0} + \ket{\lambda_1}}{\sqrt{2}},
 \label{eq:ket_plus}
\end{alignat}
and the minimum value is given by $t = |\lambda_0 + \lambda_1| / 2$.

In the computation of the upper bound $\uP{N}$, only the relation $p_{i-1} \le f(p_i)$
between the success probabilities of adjacent processes $\mE_{i-1}$ and $\mE_i$ is considered,
and no other constraints are taken into account.
Nevertheless, when $\Lambda_0$ and $\Lambda_1$ are unitary channels,
$\uP{N}$ is equal to the maximum average success probability $P^{(N)}$, as shown in the following theorem.
\begin{thm} \label{thm:unamb_opt}
 In the change point problem under unambiguous discrimination of unitary channels,
 $P^{(N)} = \uP{N}$ holds.
\end{thm}
\begin{proof}
 We consider a discrimination strategy represented as in Fig.~\ref{fig:discrimination},
 where $k$ denotes the change point.
 $\rho$ is an input state, $\mD^{(2)}, \dots, \mD^{(N)}$ are channels,
 and $\{ \Pi_m \}_{m=0}^{N+1}$ is a measurement.
 We assume, without the loss of generality, that $\mD^{(2)}, \dots, \mD^{(N)}$ are unitary channels.
 Each $\Pi_m$ for $m \in \{0,\dots,N\}$ corresponds to answering that the change point is $m$,
 while $\Pi_{N+1}$ corresponds to answering ``I don't know''.
 \begin{figure}[bt]
  \centering
  \includegraphics[scale=1.0]{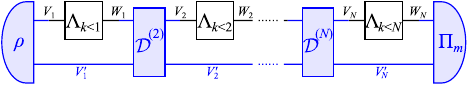}
  \caption{Discrimination in the change point problem for general quantum channels.
  Let $k$ denote the change point.  
  Any discrimination strategy is represented as a family comprising a quantum state $\rho$,
  channels $\mD^{(2)},\dots,\mD^{(N)}$, and a measurement $\{ \Pi_m \}_{m=0}^{N+1}$,
  which are represented in blue.
  All systems $\V_1,\V_2,\dots,\V_N$ have identical levels,
  and the same holds for $\W_1,\W_2,\dots,\W_N$.
  $\V'_1,\dots,\V'_N$ are auxiliary systems.}
  \label{fig:discrimination}
 \end{figure}

 Let $\ket{\psi^{(n)}_k}$ denote the state of $\W_n \ot \V'_n$ immediately after
 the $n$-th channel is applied when the change point is $k$.
 The candidates for the state just before the measurement are
 $\{ \ket{\psi^{(N)}_k} \}_{k=0}^N$, and an optimal measurement is performed on
 these candidates.
 The average success probability in this case is determined by the Gram matrix $G$
 of $\{ \ket{\psi^{(N)}_k} \}_{k=0}^N$, where the $(i+1,j+1)$-th entry of $G$ is given by
 $G_{i,j} \coloneqq \braket{\psi^{(N)}_i|\psi^{(N)}_j}$.
 The problem of finding the optimal measurement $\{ \Pi_m \}_{m=0}^{N+1}$
 for a given Gram matrix $G$ can be treated as an unambiguous discrimination problem
 for pure states, which is relatively easy to analyze.
 For simplicity, we assume that $\ket{\psi^{(N)}_0}, \dots, \ket{\psi^{(N)}_N}$
 are linearly independent (i.e., $G$ is nonsingular).
 In this case, there exist states $\ket{\tpsi_0}, \dots, \ket{\tpsi_N}$ satisfying
 $\braket{\tpsi_i|\psi^{(N)}_j} = \delta_{i,j}$, and each $\Pi_m$ with $m \in \{0,\dots,N\}$
 can be expressed as $\Pi_m = x_m \ket{\tpsi_m} \bra{\tpsi_m}$
 using a nonnegative real number $x_m$.
 The success probability when the change point is $k$ is then given by
 $\braket{\psi^{(N)}_k|\Pi_k|\psi^{(N)}_k} = x_k$.
 A necessary and sufficient condition for the existence of such a measurement is that
 $\Pi_{N+1} = \I - \sum_{m=0}^N \Pi_m$ is positive semidefinite, i.e.,
 \begin{alignat}{1}
  \I - \sum_{m=0}^N x_m \ket{\tpsi_m} \bra{\tpsi_m} &\ge \zero
  \label{eq:I_nas}
 \end{alignat}
 holds, where we write $X \ge Y$ to indicate that $X - Y$ is positive semidefinite,
 for any two matrices $X$ and $Y$.
 This is equivalent to requiring that the matrix whose $(i+1,j+1)$-th entry is
 \begin{alignat}{1}
  \bra{\psi^{(N)}_i} \left( \I - \sum_{m=0}^N x_m \ket{\tpsi_m} \bra{\tpsi_m} \right)
  \ket{\psi^{(N)}_j}
  &= G_{i,j} - x_i \delta_{i,j}
 \end{alignat}
 is positive semidefinite, i.e.,
 \begin{alignat}{1}
  G &\ge \sum_{i=0}^N x_i \ket{i} \bra{i}
  \label{eq:G_nas0}
 \end{alignat}
 holds.
 Appendix~\ref{sec:Gdep} demonstrates that it is sufficient for Eq.~\eqref{eq:G_nas0} to hold
 even when $\ket{\psi^{(N)}_0}, \dots, \ket{\psi^{(N)}_N}$ are linearly dependent.

 Now, consider the case in which a discrimination strategy exists such that
 the success probabilities for even and odd change points are $1 - tc$ and $1 - tc^{-1}$,
 respectively.
 Let
 \begin{alignat}{1}
  c_i &\coloneqq
  \begin{dcases}
   c, & i \equiv_2 0, \\
   c^{-1}, & i \equiv_2 1,
  \end{dcases}
  \label{eq:ci}
 \end{alignat}
 where $a \equiv_2 b$ means that $a - b$ is divisible by 2; then,
 from the above discussion, if Eq.~\eqref{eq:G_nas0} with $x_i = 1 - tc_i$, i.e.,
 \begin{alignat}{1}
  G \ge \sum_{i=0}^N (1 - tc_i) \ket{i}\bra{i},
  \label{eq:G_nas}
 \end{alignat}
 holds, then such a discrimination strategy exists.
 Moreover, the average success probability using this strategy is
 \begin{alignat}{1}
  \frac{1}{N+1} \sum_{l=0}^N (1 - t c_l)
  &= 1 - t \cdot \frac{(N+2)c + Nc^{-1}}{2(N+1)} = \uP{N}.
 \end{alignat}
 Therefore, it suffices to show that there exist $\rho$ and $\mD^{(2)},\dots,\mD^{(N)}$
 such that the Gram matrix $G$ satisfies Eq.~\eqref{eq:G_nas}.

 To achieve the average success probability $\uP{N}$, it is desirable to discriminate between
 the two processes $\mE_{n-1}$ and $\mE_n$ as accurately as possible
 for each $n \in \{ 1,\dots,N \}$.
 The only difference between $\mE_{n-1}$ and $\mE_n$ is whether the $n$-th channel
 is $\Lambda_0$ or $\Lambda_1$.
 Thus, to discriminate accurately between the cases in which the change point $k$ is $n - 1$ or $n$,
 we need to choose a state $\ket{\tvarphi_n}$ of the system $\V_n \ot \V'_n$
 such that the absolute value of the inner product between
 $(U_0 \ot \I) \ket{\tvarphi_n}$ and $(U_1 \ot \I) \ket{\tvarphi_n}$ is minimized.
 As discussed previously, $\ket{\tvarphi_n}$ may be chosen as the product state
 $\ket{+} \ot \ket{0}$.
 Let
 \begin{alignat}{1}
  u &\coloneqq \frac{\braket{+|U_0^\dagger U_1|+}}{t};
  \label{eq:u}
 \end{alignat}
 then, we have $\braket{\psi^{(n)}_n|\psi^{(n)}_{n-1}} = tu$.
 Note that when the change point $k$ is neither $n - 1$ nor $n$,
 the state of the system $\V_n \ot \V'_n$ at the $n$-th step does not need to be
 a product state.
 Also, since the inner product is preserved before and after applying the unitary channel $\mD^{(n)}$,
 and $\Lambda_1$ is applied at the $n$-th step when $k < n$,
 $\braket{\psi^{(n)}_i|\psi^{(n)}_j} = \braket{\psi^{(n-1)}_i|\psi^{(n-1)}_j}$
 holds for any $n \ge 2$ and $i, j < n$.

 We now consider the optimization of $\rho$ and $\mD^{(2)},\dots,\mD^{(N)}$ under these conditions.
 While the optimization is nontrivial, a heuristic solution can be found.
 As detailed in Appendix~\ref{sec:unamb_opt}, there exist $\rho$ and $\mD^{(2)},\dots,\mD^{(N)}$
 such that
 \begin{alignat}{1}
  G &= \sum_{i=0}^N (1 - tc_i) \ket{i}\bra{i} + t \ket{\mu} \bra{\mu},
  \label{eq:GN} \\
  \ket{\mu} &\coloneqq \sum_{i=0}^N u^i \sqrt{c_i} \ket{i}.
 \end{alignat}
 Since $t \ket{\mu} \bra{\mu}$ is positive semidefinite, the Gram matrix $G$ satisfies
 Eq.~\eqref{eq:G_nas}.
\end{proof}

\section{Lower bound on the maximum average success probability} \label{sec:LowerBound}

A lower bound on $P^{(N)}$ can be obtained by restricting to a specific class of
unambiguous discrimination strategies and evaluating the maximum average success probability
within that class.
Here, to enable efficient optimization, we consider a class of local discrimination strategies.
In Ref.~\cite{Sen-Cal-Mun-2017}, nonadaptive local discrimination was studied
in the context of detecting change points in quantum states.
In contrast, we consider adaptive local discrimination based on Bayesian updating.
In this strategy, an initial state $\rho_1$ is prepared and input into
the first-step channel $\Lambda_{k < 1}$, followed by a measurement $\{ \Pi^{(1)}_m \}_m$
on the output (see Fig.~\ref{fig:discrimination_Bayes}).
Next, a state $\rho_2$ is prepared and input into the second-step channel $\Lambda_{k < 2}$,
followed by a measurement $\{ \Pi^{(2)}_m \}_m$.
This process continues similarly.
For each step $n \ge 2$, the state $\rho_n$ and the measurement $\{ \Pi^{(n)}_m \}_m$ are
allowed to adapt based on the outcome of the previous measurement $\{ \Pi^{(n-1)}_m \}_m$.
Each measurement $\{ \Pi^{(n)}_m \}_m$ returns one of three outcomes: 0, 1, or \question,
indicating an inconclusive result.
To ensure error-free discrimination, each measurement $\{ \Pi^{(n)}_m \}_m$ must satisfy
the unambiguous condition, meaning that if $k < n$ (resp. $k \ge n$) holds,
then the outcome must not be 1 (resp. 0).
The maximum average success probability under this restricted class of discrimination strategies,
denoted by $\lP{N}$, can be obtained by solving a simple optimization problem
(see Appendix~\ref{sec:Bayes}).
The value $\lP{N}$ provides a lower bound on $P^{(N)}$.
Note that such local discrimination strategies offer the advantage of detecting
the change point immediately, unlike global discrimination strategies.
\begin{figure}[bt]
 \centering
 \includegraphics[scale=1.0]{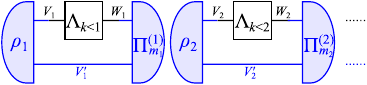}
 \caption{Adaptive local discrimination based on Bayesian updating. Each measurement $\{ \Pi^{(n)}_m \}_m$ is unambiguous, and both the state $\rho_n$ $(n \ge 2)$ and the measurement $\{ \Pi^{(n)}_m \}_m$ are adaptively updated based on the outcome of the previous measurement $\{ \Pi^{(n-1)}_m \}_m$.}
 \label{fig:discrimination_Bayes}
\end{figure}

Since the computational cost of numerically determining the maximum average success probability $P^{(N)}$
for general channels scales exponentially with $N$,
it becomes infeasible to compute $P^{(N)}$ unless $N$ is sufficiently small.
In contrast, both $\lP{N}$ and $\uP{N}$ can be computed with significantly lower complexity,
typically $O(N)$.
If the difference between $\lP{N}$ and $\uP{N}$ is sufficiently small, then
these bounds can serve as accurate approximations of $P^{(N)}$.

\section{Numerical experiments}

Figure~\ref{fig:result} presents the evaluated results of the maximum average success probability $P^{(N)}$,
its upper bound $\uP{N}$, and lower bound $\lP{N}$.
Panels~(a) and (b) illustrate the scenario in which both $\Lambda_0$ and $\Lambda_1$ are
unitary channels, while panels~(c) and (d) correspond to the case in which $\Lambda_0$ is
the identity channel and $\Lambda_1$ is an amplitude damping channel.
In the former case, as shown in Theorem~\ref{thm:unamb_opt}, $\uP{N}$ equals $P^{(N)}$.
In the latter case, $\uP{N}$ is generally greater than $P^{(N)}$, but the difference
between $\uP{N}$ and $\lP{N}$ was at most approximately $0.05$ when varying
the damping parameter in the range $1 \le N \le 10$.
In panel~(b), $P^{(N)}$ oscillates depending on whether $N$ is even or odd,
as shown in Eq.~\eqref{eq:uP_oscillate}.
According to Eq.~\eqref{eq:uP_uni_unamb}, when $N$ is odd, $P^{(N)} = 1 - t$ holds, and
$P^{(N)}$ converges to $1 - t$ in the limit as $N \to \infty$.
Furthermore, additional numerical experiments were conducted for other channels,
confirming that the upper and lower bounds yield moderate results.
\begin{figure}[bt]
 \centering
 \includegraphics[scale=1.0]{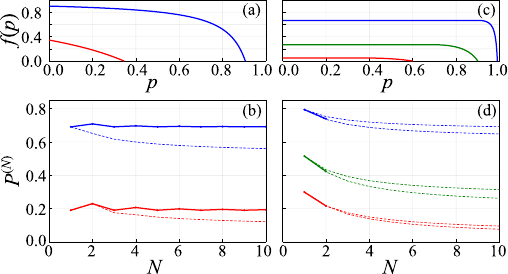}
 \caption{(a)(c) Map $f$, and (b)(d) the maximum average success probability $P^{(N)}$ (solid lines),
 along with its upper bound $\uP{N}$ and lower bound $\lP{N}$ (dashed lines).
 Panels~(a) and (b) correspond to the case in which both $\Lambda_0$ and $\Lambda_1$ are
 unitary channels, with $\lambda_1 / \lambda_0$ equal to $\exp(0.4 \pi\sqrt{-1})$ (red)
 and $\exp(0.8 \pi\sqrt{-1})$ (blue), respectively.
 The value $t = |\lambda_1 / \lambda_0 + 1| / 2$ is approximately $0.809$ and $0.309$,
 respectively.
 Since $f$ is an involution, Eq.~\eqref{eq:uP_oscillate} holds.
 Panels~(c) and (d) correspond to the case in which $\Lambda_0$ is the identity channel
 and $\Lambda_1$ is the amplitude damping channel, with damping parameters of $0.6$ (red),
 $0.9$ (green), and $0.99$ (blue).
 In panel~(d), due to computational limitations, $P^{(N)}$ was only calculated for $N = 1$
 and $N = 2$, and thus only these values are plotted.}
 \label{fig:result}
\end{figure}

\section{Conclusion}

This paper has presented a method for deriving upper and lower bounds on
the maximum average success probability in the unambiguous discrimination of change points
involving two arbitrary quantum channels.
The upper bound is derived solely from the distinguishability constraints between the two channels
and is given as the solution to the optimization problem in Eq.~\eqref{eq:uPN}.
In particular, when both channels are unitary, we have shown that this upper bound is equal to
the maximum average success probability, and its value can be analytically expressed
in terms of the channel sequence length $N$ and the distinguishability limit $1 - t$
of the two channels.
The lower bound is obtained by considering adaptive local discrimination
based on Bayesian updating.
Numerical experiments confirm that the gap between the upper and lower bounds
remains moderate.

\section*{Acknowledgment}

We thank for O.~Hirota, M.~Sohma, T.~S.~Usuda, and K.~Kato for insightful discussions.
This work was supported by the Air Force Office of Scientific Research under
award number FA2386-22-1-4056.

\appendix

\section{Notation}

Let $\Natural$ be the set of all nonnegative integers.
$\I_\V$ denotes the identity matrix on $\V$.
Let $\Pos(\V,\W)$ be the set of all completely positive maps,
which we call single-step processes, from a system $\V$ to a system $\W$.

A single-step process $f \in \Pos(\V,\W)$ is depicted as
\begin{alignat}{1}
 \includegraphics[scale=1.0]{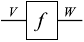} ~\raisebox{.1em}{,}
 \label{eq:single_step_process_pdf}
\end{alignat}
where labeled wires denote systems.
A one-level system is represented by the absence of a wire in the diagram.
The sequential or parallel composition of two or more single-step processes is allowed.
We consider the concatenation of $N$ single-step processes, denoted by an $N$-step process,
$\{ \mE^{(t)} \in \Pos(\W'_{t-1} \ot \V_t, \W'_t \ot \W_t) \}_{t=1}^N$ such as
\begin{alignat}{1}
 \includegraphics[scale=1.0]{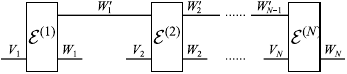} ~\raisebox{.1em}{,}
 \label{eq:comb_pdf}
\end{alignat}
where $W'_0$ and $W'_N$ are one-level.
This $N$-step process is denoted by $\mE^{(N)} \ast \cdots \ast \mE^{(1)}$,
where $\ast$ denotes the concatenation.
We call $\mE^{(N)} \ast \cdots \ast \mE^{(1)}$ a quantum comb \cite{Chi-Dar-Per-2008}
if $\mE^{(1)},\dots,\mE^{(N)}$ are quantum channels.

The change point problem addressed in this paper is the problem
of discriminating $N+1$ quantum combs $\mE_0,\dots,\mE_N$,
where $\mE_k \coloneqq \Lambda_{k < N} \ast \cdots \ast \Lambda_{k < 1}$.
$\Lambda_{k < n}$ is a channel $\Lambda_0$ if $k < n$ holds,
and a channel $\Lambda_1$ otherwise.
Each $\mE_k$ is depicted by
\begin{alignat}{1}
 \includegraphics[scale=1.0]{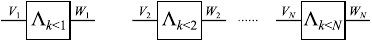} ~\raisebox{.1em}{.}
 \label{eq:comb_channel_pdf}
\end{alignat}
We consider the concatenation of a $(N+1)$-step process
$\mD_m \coloneqq \Pi_m \ast \mD^{(N)} \ast \cdots \ast \mD^{(1)}$
represented by
\begin{alignat}{1}
 \includegraphics[scale=1.0]{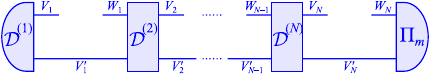}
 \label{eq:comb2_pdf}
\end{alignat}
with a quantum comb $\mE_k$ of Eq.~\eqref{eq:comb_channel_pdf}.
We refer to $\{ \mD_m \}_{m=0}^M$ as a tester if $\mD^{(1)}$ is a state,
$\mD^{(2)},\dots,\mD^{(N)}$ are quantum channels,
and $\{ \Pi_m \}_{m=0}^M$ is a quantum measurement.
$M$ is equal to $N$ for minimum-error strategies and $N + 1$ for unambiguous strategies.
The process represented by Fig.~\ref{fig:discrimination}, where $\rho = \mD^{(1)}$,
is the concatenation of $\mD_m$ of Eq.~\eqref{eq:comb2_pdf} with $\mE_k$.

\section{Proof of Theorem~\ref{thm:R}} \label{sec:R}

\subsection{Preliminaries}

The map $f \colon [0,\op] \to [0,\op]$ in Sec.~\ref{sec:UB} is
nonnegative, monotonically nonincreasing, concave, and continuous.
For each $n \in \Natural$ and $p \in [0,\op]$, let $Q'_n(p)$ be the optimal value
of the following optimization problem:
\begin{alignat}{1}
 \begin{array}{ll}
  \text{maximize} & \displaystyle \sum_{i=0}^n p_i \\
  \text{subject~to} & \displaystyle p_n = p, ~ p_{i-1} \in [0,f(p_i)]
   ~(\forall i \in \{1,\dots,n\}). \\
 \end{array}
 \label{prob:Q_}
\end{alignat}
Also, let
\begin{alignat}{1}
 Q_n(p) &\coloneqq \max_{q \in [0,p]} Q'_n(q), \label{eq:Qn_max} \\
 R_n(p) &\coloneqq \sum_{i=0}^n f^{(i)}(p), \nonumber \\
 r_n &\coloneqq \max \mathop{\argmax}_{p \in [0,\op]} R_n(p), \nonumber \\
 R^\opt_n &\coloneqq \max_{p \in [0,p]} R_n(p) = R_n(r_n);
\end{alignat}
then, from $\uP{N} = Q_n(\op) / (N+1)$,
Theorem~\ref{thm:R} is equivalent to $Q_N(\op) = R^\opt_N$.
From the definition of $Q'_n(p)$, we have, for any $n \in \Natural$ with $n \ge 1$ and
$p \in [0,\op]$,
\begin{alignat}{1}
 Q'_n(p) &= p + Q_{n-1}[f(p)].
 \label{eq:Qn_pQ}
\end{alignat}

For each $k \in \Natural$ and $y \in [0,\op]$, let
\begin{alignat}{1}
 \cP_{k,y} \coloneqq \left\{ p \relmid f^{(k)}(p) = y \right\}.
\end{alignat}
Let
\begin{alignat}{1}
 f^{(-k)}(y) &\coloneqq
 \begin{dcases}
  \max \cP_{k,y}, & k ~\text{is~even}, \\
  \min \cP_{k,y}, & k ~\text{is~odd} \\
 \end{dcases}
\end{alignat}
if $\cP_{k,y}$ is nonempty.
Note that if $\cP_{k,y}$ is nonempty, then $f^{(k)}[f^{(-k)}(y)] = y$ holds.
Also, let
\begin{alignat}{1}
 t_n &\coloneqq
 \begin{dcases}
  f^{(-1)}(r_{n-1}), & \text{$\cP_{1,r_{n-1}}$ is nonempty}, \\
  0, & \text{otherwise} \\
 \end{dcases}
 \label{eq:tn}
\end{alignat}
for $1 \le n \in \Natural$ and $t_0 \coloneqq 0$.
$\cP_{1,y}$ being nonempty is equivalent to $f(\op) \le y \le f(0)$,
and thus is equivalent to the existence of $p,q \in [0,\op]$ satisfying $f(p) \le y \le f(q)$.

The feasible set of Problem~\eqref{prob:Q_} is convex.
Indeed, for any two feasible solutions $\{ p_i \}_{i=0}^n$ and $\{ p'_i \}_{i=0}^n$
and any real number $c$ with $0 \le c \le 1$,
$\{ q_i \coloneqq c p_i + (1-c) p'_i \}_{i=0}^n$ satisfies
$q_n = p$ and $q_{i-1} \in [0, c f(p_i) + (1-c) f(p'_i)] \subseteq [0, f(q_i)]$
for each $i \in \{0,\dots,n\}$.
Thus, $Q'_n$ is concave over the interval $[0,\op]$.
Also, $Q'_0(p) = R_0(p) = p$, $r_0 = \op$, and $Q'_1(p) = R_1(p) = p + f(p)$ hold.
$Q'_1 = R_1$ is monotonically nondecreasing in $[0,r_1]$ and
monotonically nonincreasing in $[r_1,\op]$.
We have, for any $p,q \in [0,\op]$,
\begin{alignat}{3}
 p \le q &\quad\Rightarrow\quad f(p) \ge f(q),
 \label{eq:pqfpfq} \\
 f(p) > f(q) &\quad\Rightarrow\quad p < q.
 \label{eq:pqfpfq_inv}
\end{alignat}
When $\cP_{1,p}$ is nonempty, we have
\begin{alignat}{3}
 p \ge f(q) &\quad\Rightarrow\quad f^{(-1)}(p) \le q.
 \label{eq:fpq}
\end{alignat}

\subsection{Lemmas}

\begin{lemma} \label{lemma:rnfrn}
 $r_n \ge f^{(n+1)}(r_n)$ holds for each $n \in \Natural$.
\end{lemma}
\begin{proof}
 \begin{alignat}{1}
  r_n - f^{(n+1)}(r_n) &= \sum_{i=0}^n f^{(i)}(r_n) - \sum_{i=1}^{n+1} f^{(i)}(r_n)
  \nonumber \\
  &= R_n(r_n) - R_n[f(r_n)] \nonumber \\
  &\ge 0.
 \end{alignat}
\end{proof}

\begin{lemma} \label{lemma:rnfop}
 $r_n \ge f(\op)$ holds for each $n \in \Natural$.
\end{lemma}
\begin{proof}
 $r_n \ge f^{(n+1)}(r_n) = f[f^{(n)}(r_n)] \ge f(\op)$.
\end{proof}

\begin{lemma} \label{lemma:f2r}
 For each $n \in \Natural$ with $n \ge 2$, if $r_{n-2} \ge r_1$ holds,
 then $f^{(2)}(r_n) \le r_{n-2}$ holds.
\end{lemma}
\begin{proof}
 The proof proceeds by contradiction.
 Assume $f^{(2)}(r_n) > r_{n-2}$.
 From $f(\op) \le r_{n-2} < f^{(2)}(r_n)$, $\cP_{1,r_{n-2}}$ is nonempty.
 Also, we have $f(r_n) < f^{(-1)}(r_{n-2}) \le f(r_1)$,
 where the first inequality follows by applying Eq.~\eqref{eq:pqfpfq_inv} to $r_{n-2} < f^{(2)}(r_n)$
 and the second inequality follows by applying Eq.~\eqref{eq:fpq} to
 $r_{n-2} \ge r_1 \ge f^{(2)}(r_1)$, which follows by substituting $n = 1$
 into Lemma~\ref{lemma:rnfrn}.
 Thus, $\cP_{1,f^{(-1)}(r_{n-2})}$ is nonempty.
 From $\cP_{1,f^{(-1)}(r_{n-2})} \subseteq \cP_{2,r_{n-2}}$, $\cP_{2,r_{n-2}}$ is also nonempty.
 Applying Eq.~\eqref{eq:pqfpfq_inv} twice to $f^{(2)}(r_n) > r_{n-2}$
 gives $f^{(-2)}(r_{n-2}) < r_n$.
 Moreover, $R_{n-2}(r_{n-2}) = R^\opt_{n-2} > R_{n-2}[f^{(2)}(r_n)]$ holds.
 Indeed, if $R_{n-2}[f^{(2)}(r_n)] = R^\opt_{n-2}$,
 i.e., $f^{(2)}(r_n) \in \mathop{\argmax}_{p \in [0,\op]} R_{n-2}(p)$, holds,
 then $r_{n-2} \ge f^{(2)}(r_n)$ must hold, which contradicts the assumption.
 Thus, from $R_n[f^{(-2)}(r_{n-2})] \le R^\opt_n = R_n(r_n)$, we have
 \begin{alignat}{1}
  R_n[f^{(-2)}(r_{n-2})] - R_{n-2}(r_{n-2})
  &< R_n(r_n) - R_{n-2}[f^{(2)}(r_n)]. \nonumber \\
  \label{eq:f2r_R1}
 \end{alignat}
 We can rewrite the left-hand side as
 \begin{alignat}{1}
  \lefteqn{ R_n[f^{(-2)}(r_{n-2})] - R_{n-2}(r_{n-2}) } \nonumber \\
  &= \sum_{i=0}^n f^{(i)}[f^{(-2)}(r_{n-2})] - \sum_{i=2}^n f^{(i)}[f^{(-2)}(r_{n-2})] \nonumber \\
  &= \sum_{i=0}^1 f^{(i)}[f^{(-2)}(r_{n-2})] \nonumber \\
  &= R_1[f^{(-2)}(r_{n-2})].
  \label{eq:f2r_R1_Rn}
 \end{alignat}
 Also, we can rewrite the right-hand side of Eq.~\eqref{eq:f2r_R1} as
 \begin{alignat}{1}
  R_n(r_n) - R_{n-2}[f^{(2)}(r_n)] &= \sum_{i=0}^n f^{(i)}(r_n) - \sum_{i=0}^{n-2} f^{(i)}[f^{(2)}(r_n)]
  \nonumber \\
  &= \sum_{i=0}^1 f^{(i)}(r_n) \nonumber \\
  &= R_1(r_n).
  \label{eq:f2p_Rn_Rn2}
 \end{alignat}
 Thus, we have $R_1[f^{(-2)}(r_{n-2})] < R_1(r_n)$.
 Therefore, from $f^{(-2)}(r_{n-2}) < r_n$ and the fact that
 $R_1$ is concave and monotonically nonincreasing in $[r_1,\op]$, we obtain
 \begin{alignat}{1}
  f^{(-2)}(r_{n-2}) &< r_1.
  \label{eq:f2rn2_le_r1}
 \end{alignat}
 Applying Eq.~\eqref{eq:pqfpfq} twice to this inequality and using Lemma~\ref{lemma:rnfrn}
 with the substitution $n = 1$ yields $r_{n-2} \le f^{(2)}(r_1) \le r_1$,
 and thus $r_{n-2} = f^{(2)}(r_1)$ holds from the assumption $r_1 \le r_{n-2}$.
 Thus, we obtain $f^{(-2)}(r_{n-2}) = \max \{ p \mid f^{(2)}(p) = r_{n-2} \} \ge r_1$,
 which contradicts Eq.~\eqref{eq:f2rn2_le_r1}.
\end{proof}

\begin{lemma} \label{lemma:rnr1}
 $r_n \ge r_1$ holds for each $n \in \Natural$.
\end{lemma}
\begin{proof}
 Since $r_0 = \op \ge r_1$ and $r_1 \ge r_1$ hold, it suffices to consider the case
 of $n \ge 2$.
 For any $k \in \Natural$ with $k \ge 1$, it is sufficient to show
 $r_{2k} \ge r_{2k-1}$ and $r_{2k+1} \ge r_1$ (in which case, for example,
 $r_6 \ge r_5 \ge r_1$ holds).

 We first show $r_{2k} \ge r_{2k-1}$.
 We have
 \begin{alignat}{1}
  f^{(2k)}(r_{2k}) &= R_{2k}(r_{2k}) - R_{2k-1}(r_{2k}) \nonumber \\
  &\ge R_{2k}(r_{2k-1}) - R_{2k-1}(r_{2k-1}) \nonumber \\
  &= f^{(2k)}(r_{2k-1}),
 \end{alignat}
 where the inequality follows from
 $R_{2k}(r_{2k}) \ge R_{2k}(r_{2k-1})$ and $R_{2k-1}(r_{2k}) \le R_{2k-1}(r_{2k-1})$.
 If $f^{(2k)}(r_{2k}) = f^{(2k)}(r_{2k-1})$ holds, then
 $R_{2k}(r_{2k}) = R_{2k}(r_{2k-1})$ (from which $r_{2k} \ge r_{2k-1}$ follows)
 and $R_{2k-1}(r_{2k}) = R_{2k-1}(r_{2k-1})$ (from which $r_{2k} \le r_{2k-1}$ follows) hold,
 and thus $r_{2k} = r_{2k-1}$ holds.
 If $f^{(2k)}(r_{2k}) > f^{(2k)}(r_{2k-1})$ holds, then applying Eq.~\eqref{eq:pqfpfq_inv}
 $2k$ times yields $r_{2k} > r_{2k-1}$.

 Next, we show $r_{2k+1} \ge r_1$.
 It is sufficient to assume the inductive hypothesis
 $r_{2k-1} \ge r_{2k-3} \ge \dots \ge r_3 \ge r_1$ and prove
 $r_{2k+1} \ge r_{2k-1}$.
 Let $s_\pm \coloneqq f^{(2k)}(r_{2k \pm 1})$; then, we have
 \begin{alignat}{1}
  R_1(s_+) &= R_{2k+1}(r_{2k+1}) - R_{2k-1}(r_{2k+1}) \nonumber \\
  &\ge R_{2k+1}(r_{2k-1}) - R_{2k-1}(r_{2k-1}) \nonumber \\
  &= R_1(s_-),
  \label{eq:rnr1_spm}
 \end{alignat}
 where the inequality follows from $R_{2k+1}(r_{2k+1}) \ge R_{2k+1}(r_{2k-1})$ and
 $R_{2k-1}(r_{2k+1}) \le R_{2k-1}(r_{2k-1})$.
 On the other hand, by repeatedly applying Lemma~\ref{lemma:f2r} and using Eq.~\eqref{eq:pqfpfq},
 we obtain $s_+ \le f^{(2k-2)}(r_{2k-1}) \le \dots \le r_1$ and
 $s_- \le f^{(2k-2)}(r_{2k-3}) \le \dots \le f^{(2)}(r_1) \le r_1$
 (the last inequality follows by substituting $n = 1$ into Lemma~\ref{lemma:rnfrn}),
 and thus both $s_+$ and $s_-$ are less than or equal to $r_1$.
 Combining this with Eq.~\eqref{eq:rnr1_spm} [i.e., $R_1(s_+) \ge R_1(s_-)$] and
 the fact that $R_1$ is monotonically nondecreasing on $[0, r_1]$,
 we conclude that $s_+ \ge s_-$.
 When $s_+ = s_-$, $R_1(s_+) = R_1(s_-)$ holds, and thus from Eq.~\eqref{eq:rnr1_spm},
 we obtain $R_{2k+1}(r_{2k+1}) = R_{2k+1}(r_{2k-1})$
 (from which $r_{2k+1} \ge r_{2k-1}$ follows) and
 $R_{2k-1}(r_{2k+1}) = R_{2k-1}(r_{2k-1})$ (from which $r_{2k+1} \le r_{2k-1}$ follows),
 which gives $r_{2k+1} = r_{2k-1}$.
 When $s_+ > s_-$, applying Eq.~\eqref{eq:pqfpfq_inv} $2k$ times yields
 $r_{2k+1} > r_{2k-1}$.
\end{proof}

\begin{cor} \label{cor:f2r}
 $f^{(2)}(r_n) \le r_{n-2}$ holds for each $n \in \Natural$ with $n \ge 2$.
\end{cor}
\begin{proof}
 This follows directly from Lemmas~\ref{lemma:f2r} and \ref{lemma:rnr1}.
\end{proof}

\begin{lemma} \label{lemma:tnrn}
 $t_n \le r_n$ holds for each $n \in \Natural$.
\end{lemma}
\begin{proof}
 The case $n = 0$ follows directly from $t_0 = 0 \le \op = r_0$.
 Now we consider the case $n \ge 1$.
 If $\cP_{1,r_{n-1}}$ is empty, then $t_n = 0 \le r_n$ holds.
 Otherwise, we have
 \begin{alignat}{1}
  t_n &= f^{(-1)}(r_{n-1}) \nonumber \\
  &= R_n[f^{(-1)}(r_{n-1})] - R_{n-1}(r_{n-1}) \nonumber \\
  &\le R_n(r_n) - R_{n-1}[f(r_n)] \nonumber \\
  &= r_n,
 \end{alignat}
 where the inequality follows from
 $R_n[f^{(-1)}(r_{n-1})] \le R_n(r_n)$ and $R_{n-1}(r_{n-1}) \ge R_{n-1}[f(r_n)]$.
\end{proof}

\begin{lemma} \label{lemma:fptn}
 For each $n \in \Natural$ with $n \ge 1$, $f(p) \ge t_{n-1}$ holds for any $p \in [0, r_n]$.
\end{lemma}
\begin{proof}
 If $\cP_{1,r_{n-2}}$ is empty, then this follows directly from $t_{n-1} = 0$.
 Otherwise, from $p \le r_n$ and Eq.~\eqref{eq:pqfpfq}, we have $f(p) \ge f(r_n)$.
 Applying Eq.~\eqref{eq:fpq} to $f^{(2)}(r_n) \le r_{n-2}$,
 which is obtained from Corollary~\ref{cor:f2r}, gives $f(r_n) \ge f^{(-1)}(r_{n-2}) = t_{n-1}$.
 Therefore, we obtain $f(p) \ge f(r_n) \ge t_{n-1}$.
\end{proof}

\subsection{Proof}

We are now ready to prove Theorem~\ref{thm:R}.
The left- and right-hand sides of Eq.~\eqref{eq:uPN} are equal to
$Q_N(\op)/(N+1)$ and $R^\opt_N/(N+1)$, respectively.
Therefore, it suffices to show that $Q_n(\op) = R^\opt_n$ holds for any $n \in \Natural$.
To establish this equality, it suffices to verify the following three conditions:
\begin{enumerate}[label=(\arabic*)]
 \item $Q'_n(p) \le R^\opt_n$ holds for any $p < t_n$.
 \item $Q'_n(p) = R_n(p)$ holds for any $t_n \le p \le r_n$.
 \item $Q'_n(p) \le R^\opt_n$ holds for any $p > r_n$.
\end{enumerate}
Note that $t_n \le r_n$ holds by Lemma~\ref{lemma:tnrn}.
In this case, $Q'_n(p)$ attains its maximum value $R^\opt_n$ at $p = r_n$, and thus,
from Eq.~\eqref{eq:Qn_max}, we obtain
\begin{alignat}{1}
 Q_n(p) &= \max_{q \in [0,p]} Q'_n(q) = R^\opt_n
 \label{eq:Qn_proof}
\end{alignat}
for any $p \ge r_n$, and thus $Q_n(\op) = R^\opt_n$.
Furthermore, if Conditions~(1)--(3) hold, then from the concavity of $Q'_n$
and Eq.~\eqref{eq:Qn_max}, we obtain
\begin{alignat}{1}
 Q_n(p) &=
 \begin{dcases}
  Q'_n(p), & p < t_n, \\
  R_n(p), & t_n \le p \le r_n, \\
  R^\opt_n, & p > r_n.
 \end{dcases}
 \label{eq:Qn_proof0}
\end{alignat}

For $n = 0$ and $n = 1$, we have $Q'_n(p) = R_n(p)$, so Conditions~(1)--(3) hold.
In what follows, we consider the case $n \ge 2$.
To prove the result by induction, we assume that Conditions~(1)--(3) hold for $n-1$ and $n-2$.

Condition~(1): Consider the case $p < t_n$.
Since $t_n > 0$ holds, we have $f(t_n) = r_{n-1}$ from Eq.~\eqref{eq:tn}, and thus
\begin{alignat}{1}
 Q'_n(p) &= p + Q_{n-1}[f(p)] \nonumber \\
 &= p + R^\opt_{n-1} \nonumber \\
 &< t_n + R_{n-1}(r_{n-1}) \nonumber \\
 &= R_n(t_n) \nonumber \\
 &\le R^\opt_n,
 \label{eq:Qn_1}
\end{alignat}
where the first equality follows from Eq.~\eqref{eq:Qn_pQ}.
The second equality follows from $p < t_n$ and Eq.~\eqref{eq:pqfpfq},
which imply $f(p) \ge f(t_n)$, together with $f(t_n) = r_{n-1}$ and Eq.~\eqref{eq:Qn_proof}.

Condition~(2): Consider the case $t_n \le p$ and $f(p) \ge t_{n-1}$.
[Note that by Lemma~\ref{lemma:fptn}, the case $t_n \le p \le r_n$ is a special case of this.
This broader setting facilitates the subsequent proof of Condition~(3)].
From $p \ge t_n$, we have $f(p) \le r_{n-1}$.
[Indeed, if $\cP_{1,r_{n-1}}$ is empty, then either $r_{n-1} < f(\op)$ or $f(0) < r_{n-1}$ holds.
Since $r_{n-1} \ge f(\op)$ holds by Lemma~\ref{lemma:rnfop}, $f(0) < r_{n-1}$ must hold.
Therefore, $f(p) \le f(0) < r_{n-1}$ holds.
Otherwise, applying Eq.~\eqref{eq:pqfpfq} to $p \ge t_n$ yields $f(p) \le f(t_n)$,
and since $f(t_n) = r_{n-1}$ holds by Eq.~\eqref{eq:tn}, we obtain $f(p) \le r_{n-1}$.]
Thus, $t_{n-1} \le f(p) \le r_{n-1}$ holds, and from Eq.~\eqref{eq:Qn_proof0},
we have $Q_{n-1}[f(p)] = R_{n-1}[f(p)]$.
Therefore, we obtain
\begin{alignat}{1}
 Q'_n(p) &= p + Q_{n-1}[f(p)] \nonumber \\
 &= p + R_{n-1}[f(p)] \nonumber \\
 &= R_n(p),
 \label{eq:Qn_2}
\end{alignat}
where the first line follows from Eq.~\eqref{eq:Qn_pQ}.

Condition~(3): Consider the case $p > r_n$.
If $f(p) \ge t_{n-1}$ holds, then Eq.~\eqref{eq:Qn_2} holds, and thus
$Q'_n(p) = R_n(p) \le R^\opt_n$ holds.
In what follows, we consider the case $f(p) < t_{n-1}$.
From $t_{n-1} > 0$ and Eq.~\eqref{eq:tn}, we have $t_{n-1} = f^{(-1)}(r_{n-2})$.
Let $S(p) \coloneqq p + f(p) + R^\opt_{n-2}$; then, we have
\begin{alignat}{1}
 Q'_n(p) &= p + Q_{n-1}[f(p)] \nonumber \\
 &= p + Q'_{n-1}[f(p)] \nonumber \\
 &= p + f(p) + Q_{n-2}[f^{(2)}(p)] \nonumber \\
 &= S(p),
 \label{eq:QnSp}
\end{alignat}
where the first and third lines follow from Eq.~\eqref{eq:Qn_pQ}, and
the second line follows from $f(p) < t_{n-1}$ and Eq.~\eqref{eq:Qn_proof0},
which implies $Q_{n-1}[f(p)] = Q'_{n-1}[f(p)]$.
In the last line, the inequality $f^{(2)}(p) \ge r_{n-2}$ follows from
$f(p) < t_{n-1}$ and Eq.~\eqref{eq:pqfpfq}, which yield $f^{(2)}(p) \ge f(t_{n-1})$,
together with $f(t_{n-1}) = r_{n-2}$.
Then, from Eq.~\eqref{eq:Qn_proof}, we obtain $Q_{n-2}[f^{(2)}(p)] = R^\opt_{n-2}$.
We distinguish between the cases $f^{(2)}(p) = r_{n-2}$ and $f^{(2)}(p) > r_{n-2}$.
When $f^{(2)}(p) = r_{n-2}$, we have $S(p) = R_n(p)$ from $R^\opt_{n-2} = R_{n-2}[f^{(2)}(p)]$.
Thus, we have, from Eq.~\eqref{eq:QnSp},
\begin{alignat}{1}
 Q'_n(p) &= S(p) = R_n(p) \le R^\opt_n.
\end{alignat}
In what follows, consider the case $f^{(2)}(p) > r_{n-2}$.
From Corollary~\ref{cor:f2r}, we have $f^{(2)}(r_n) \le r_{n-2}$.
When $f^{(2)}(r_n) = r_{n-2}$, we have, from Eq.~\eqref{eq:QnSp},
\begin{alignat}{1}
 Q'_n(p) &= S(p) \nonumber \\
 &= R_1(p) + R_{n-2}[f^{(2)}(r_n)] \nonumber \\
 &\le R_1(r_n) + R_{n-2}[f^{(2)}(r_n)] \nonumber \\
 &= R_n(r_n) \nonumber \\
 &= R^\opt_n,
\end{alignat}
where the inequality follows from $r_1 \le r_n < p$ (see Lemma~\ref{lemma:rnr1}) and
from the monotonic decrease of $R_1$ over the interval $[r_1, \op]$,
which lead to $R_1(p) \le R_1(r_n)$.
When $f^{(2)}(r_n) < r_{n-2}$, Eq.~\eqref{eq:pqfpfq_inv} yields
$f(r_n) > f^{(-1)}(r_{n-2}) = t_{n-1} > f(p)$.
Therefore, $\cP_{1,t_{n-1}}$ is nonempty, and thus $\gamma \coloneqq f^{(-1)}(t_{n-1})$ is
well-defined.
Applying Eq.~\eqref{eq:pqfpfq_inv} to $f(r_n) > f(\gamma) > f(p)$ gives $r_n < \gamma < p$.
Using this and the concavity of $Q'_n$, we obtain
$Q'_n(\gamma) \ge \eta Q'_n(r_n) + (1 - \eta) Q'_n(p)$, i.e.,
\begin{alignat}{1}
 \eta Q'_n(p) &\le Q'_n(\gamma) - (1 - \eta) Q'_n(r_n),
 \label{Qn_concave}
\end{alignat}
where $\eta \coloneqq (p - \gamma)/(p - r_n)$.
Since Condition~(2) holds when $p = r_n$, Eq.~\eqref{eq:Qn_2} with the substitution $p = r_n$
yields $Q'_n(r_n) = R_n(r_n) = R^\opt_n$.
Also, since $t_n \le r_n < \gamma$ and $f(\gamma) = t_{n-1}$ hold,
Condition~(2) holds when $p = \gamma$,
and Eq.~\eqref{eq:Qn_2} with the substitution $p = \gamma$
gives $Q'_n(\gamma) = R_n(\gamma)$.
Therefore, from Eq.~\eqref{Qn_concave}, we obtain
$\eta Q'_n(p) \le R_n(\gamma) - (1 - \eta) R^\opt_n \le \eta R^\opt_n$,
and dividing the first and last terms by $\eta$ yields $Q'_n(p) \le R^\opt_n$.
Consequently, in all cases, we have $Q'_n(p) \le R^\opt_n$, completing the proof of Condition~(3).

\section{Proof of Theorem~\ref{thm:unamb_opt}} \label{sec:unamb_opt}

We define $t$, $\lambda_0$, and $\lambda_1$ as described in Sec.~\ref{sec:unitary}.
We also define $c$, $\ket{+}$, and $u$ by Eqs.~\eqref{eq:c}, \eqref{eq:ket_plus}, and \eqref{eq:u},
respectively.
Let
\begin{alignat}{1}
 \tU &\coloneqq U_0^\dagger U_1, \quad
 \w \coloneqq \frac{\lambda_1}{\lambda_0};
 \label{eq:tU}
\end{alignat}
then, we have
\begin{alignat}{1}
 u &= \frac{\lambda_0 + \lambda_1}{2 \lambda_0 t} = \frac{1+\w}{2t}, \quad
 u^2 = \w, \quad
 |u| = |\w| = 1.
\end{alignat}
Since $P^{(N)} = 1 = \uP{N}$ holds when $t = 0$, and $P^{(N)} = 0 = \uP{N}$ holds when $t = 1$,
it suffices to consider the case $0 < t < 1$.

We consider a tester $\{ \mD_m \}_{m=0}^M$ whose each element
$\mD_m \coloneqq \Pi_m \ast \mD^{(N)} \ast \cdots \ast \mD^{(1)}$ is represented by
Eq.~\eqref{eq:comb2_pdf},
where $\mD^{(1)} = \rho$ is a pure state and $\mD^{(2)},\dots,\mD^{(N)}$ are unitary channels.
We assume that all systems $\V_1, \dots, \V_N$ and $\W_1, \dots, \W_N$ have the same level,
and likewise for $\V'_1, \dots, \V'_N$.
We also assume that the level of $\V'_1$, denoted by $d'$, is sufficiently large.
When the change point is $k \in \{0, \dots, N\}$, the state immediately before and after
applying the channel at step $n \in \{1, \dots, N\}$ are pure states of $\V_n \ot \V'_n$ and
$\W_n \ot \V'_n$, denoted by $\ket{\varphi^{(n)}_k}$ and $\ket{\psi^{(n)}_k}$, respectively.
These states are depicted by
\begin{widetext}
\begin{alignat}{1}
 &\includegraphics[scale=1.0]{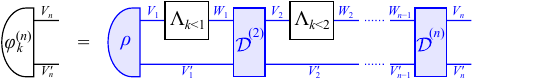} ~\raisebox{1em}{,} \nonumber \\
 &\includegraphics[scale=1.0]{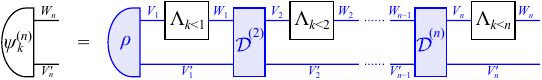} ~\raisebox{1em}{.}
 \label{eq:psi_n}
\end{alignat}
\end{widetext}
We have
\begin{alignat}{1}
 \braket{\varphi^{(n)}_i|\varphi^{(n)}_j} &= \braket{\psi^{(n-1)}_i|\psi^{(n-1)}_j},
 \quad \forall n \ge 2, ~i,j \in \{0, \dots, n\}, \nonumber \\
 \ket{\psi^{(n)}_k} &=
 \begin{dcases}
  (U_1 \ot \I_{d'}) \ket{\varphi^{(n)}_k}, & k < n, \\
  (U_0 \ot \I_{d'}) \ket{\varphi^{(n)}_{n-1}}, & k \ge n,
 \end{dcases}
\end{alignat}
where the first line follows from the preservation of inner products before and after
applying $\mD^{(n)}$.
Thus, we obtain
\begin{alignat}{1}
 \ket{\varphi^{(n)}_N} &= \ket{\varphi^{(n)}_{N-1}} = \cdots = \ket{\varphi^{(n)}_{n-1}}, \nonumber \\
 \ket{\psi^{(n)}_N} &= \ket{\psi^{(n)}_{N-1}} = \cdots = \ket{\psi^{(n)}_n}.
\end{alignat}

Let $\ket{i}_n$ (or simply $\ket{i}$) denote the $(n+1)$-dimensional column vector
whose $(i+1)$-th entry is $1$
and all others are $0$.
Define the square matrix $\tG^{(n)}$ of order $n+1$ as
\begin{alignat}{1}
 \tG^{(n)} &\coloneqq \sum_{i=0}^n (1 - tc_i) \ket{i}\bra{i} + t \ket{\mu_n} \bra{\mu_n},
 \label{eq:tG}
\end{alignat}
where $c_i$ is given by Eq.~\eqref{eq:ci}, and
\begin{alignat}{1}
 \ket{\mu_n} \coloneqq \sum_{i=0}^n u^i \sqrt{c_i} \ket{i}.
\end{alignat}
$\tG^{(N)}$ is equal to the right-hand side of Eq.~\eqref{eq:GN}.
For each $i, j \in \{0, \dots, n\}$, we have
\begin{alignat}{1}
 \braket{i|\tG^{(n)}|j} &=
 \begin{dcases}
  1, & i = j, \\
  t u^{i-j} \sqrt{c_i c_j}, & i \neq j.
 \end{dcases}
 \label{eq:iGnj}
\end{alignat}
For example, we have
\begin{alignat}{1}
 \tG^{(3)} &=
 \begin{bmatrix}
  1 & t u^{-1} & t u^{-2} c & t u^{-3} \\
  t u & 1 & t u^{-1} & t u^{-2} c^{-1} \\
  t u^2 c & t u & 1 & t u^{-1} \\
  t u^3 & t u^2 c^{-1} & t u & 1 \\
 \end{bmatrix}, \nonumber \\
 \tG^{(4)} &=
 \begin{bmatrix}
  1 & t u^{-1} & t u^{-2} c & t u^{-3} & t u^{-4} c \\
  t u & 1 & t u^{-1} & t u^{-2} c^{-1} & t u^{-3} \\
  t u^2 c & t u & 1 & t u^{-1} & t u^{-2} c \\
  t u^3 & t u^2 c^{-1} & t u & 1 & t u^{-1} \\
  t u^4 c & t u^3 & t u^2 c & t u & 1 \\
 \end{bmatrix}.
\end{alignat}
$\tG^{(n)}$ is equal to the principal submatrix of order $n+1$
obtained by deleting the $(n+2)$-th row and $(n+2)$-th column of $\tG^{(n+1)}$.

From the argument in Sec.~\ref{sec:unitary}, it is sufficient to demonstrate that
there exists a process $\mD^{(N)} \ast \cdots \ast \mD^{(1)}$ such that
the Gram matrix of $\{ \ket{\psi^{(N)}_k} \}_{k=0}^N$ is equal to $\tG^{(N)}$.
To show this, we prove that for each $n \in \{1,\dots,N\}$,
there exist a state $\rho$ and unitary channels $\mD^{(2)},\dots,\mD^{(n)}$ such that
the Gram matrix of $\{ \ket{\psi^{(n)}_k} \}_{k=0}^n$ equals $\tG^{(n)}$.

First, consider the case $n = 1$.
We take the state $\rho$ to be the pure state $\ket{+} \ot \ket{0}$,
with $\ket{0}$ being a normalized vector in $\V'_1$.
Then, $\braket{\psi^{(1)}_1|\psi^{(1)}_0} = \braket{+|\tU|+} = tu$ holds,
where the first and second equalities follow from Eqs.~\eqref{eq:tU} and \eqref{eq:u}, respectively.
Therefore, the Gram matrix of $\{ \ket{\psi^{(1)}_k} \}_{k=0}^1$ is equal to
\begin{alignat}{1}
 \tG^{(1)} &=
 \begin{bmatrix}
  1 & t u^{-1} \\
  t u & 1 \\
 \end{bmatrix}.
\end{alignat}

Next, suppose that for some $n \in \{1,\dots,N-1\}$, there exist a state $\rho$ and
unitary channels $\mD^{(2)},\dots,\mD^{(n)}$ such that
the Gram matrix of $\{ \ket{\psi^{(n)}_k} \}_{k=0}^n$ is equal to $\tG^{(n)}$.
We then show that there exists a unitary channel $\mD^{(n+1)}$ such that
the Gram matrix of $\{ \ket{\psi^{(n+1)}_k} \}_{k=0}^{n+1}$ is equal to $\tG^{(n+1)}$.
Let the channel $\mD^{(n+1)}$ be represented by a unitary matrix $D^{(n+1)}$ of order $dd'$;
then, since
\begin{alignat}{1}
 \ket{\psi^{(n+1)}_k} &= (U_1 \ot \I_{d'}) \ket{\varphi^{(n+1)}_k}
 = (U_1 \ot \I_{d'}) D^{(n+1)} \ket{\psi^{(n)}_k}
 \label{eq:psik}
\end{alignat}
holds for any $k \in \{0,\dots,n\}$, we have, for each $j, k \in \{0,\dots,n\}$,
\begin{alignat}{1}
 \braket{\psi^{(n+1)}_j|\psi^{(n+1)}_k} &= \braket{\psi^{(n)}_j|\psi^{(n)}_k}
 = \braket{j|\tG^{(n)}|k} = \braket{j|\tG^{(n+1)}|k},
\end{alignat}
i.e., the $(j+1,k+1)$ entry of the Gram matrix of $\{ \ket{\psi^{(n+1)}_k} \}_{k=0}^{n+1}$
is equal to $\braket{j|\tG^{(n+1)}|k}$.
Therefore, it suffices to prove that for each $k \in \{0,\dots,n\}$,
$\braket{\psi^{(n+1)}_{n+1}|\psi^{(n+1)}_k} = \braket{n+1|\tG^{(n+1)}|k}$ holds,
or equivalently, that there exists a unitary matrix $D^{(n+1)}$ such that $\bra{\nu} = \bra{\nu'}$,
where
\begin{alignat}{1}
 \bra{\nu} &\coloneqq
 \bra{\psi^{(n+1)}_{n+1}}
 \begin{bmatrix}
  \ket{\psi^{(n+1)}_0} & \ket{\psi^{(n+1)}_1} & \cdots & \ket{\psi^{(n+1)}_n} \\
 \end{bmatrix}, \nonumber \\
 \bra{\nu'} &\coloneqq \bra{n+1} \tG^{(n+1)}
 \begin{bmatrix}
  \ket{0} & \ket{1} & \cdots & \ket{n} \\
 \end{bmatrix}.
 \label{eq:nu}
\end{alignat}
Similar to Eq.~\eqref{eq:psik}, we have
\begin{alignat}{1}
 \ket{\psi^{(n+1)}_{n+1}} &= (U_0 \ot \I_{d'}) \ket{\varphi^{(n+1)}_n}
 = (U_0 \ot \I_{d'}) D^{(n+1)} \ket{\psi^{(n)}_n}.
 \label{eq:psin}
\end{alignat}
Let $\Xi$ denote the projection matrix onto the image of $\tG^{(n)}$.
Note that if $\tG^{(n)}$ is nonsingular, then $\Xi = \I_{n+1}$ holds.
There exists a $dd' \times (n+1)$ matrix $\Theta$ satisfying
$\Theta = \Theta \Xi$, $\Theta^\dagger \Theta = \Xi$, and
\begin{alignat}{1}
 \begin{bmatrix}
  \ket{\psi^{(n)}_0} & \ket{\psi^{(n)}_1} & \cdots & \ket{\psi^{(n)}_n} \\
 \end{bmatrix}
 &= \Theta \tG^{(n)}{}^\frac{1}{2}.
 \label{eq:Theta}
\end{alignat}
Indeed,
\begin{alignat}{1}
 \Theta &\coloneqq
 \begin{bmatrix}
  \ket{\psi^{(n)}_0} & \ket{\psi^{(n)}_1} & \cdots & \ket{\psi^{(n)}_n} \\
 \end{bmatrix}
 \tG^{(n)}{}^\frac{1}{2}{}^+
\end{alignat}
satisfies these conditions,
where ${}^+$ denotes the Moore–Penrose pseudoinverse.
To show the existence of $D^{(n+1)}$ satisfying $\bra{\nu} = \bra{\nu'}$,
we show the following two properties:
\begin{enumerate}[label=(\alph*)]
 \item If
       \begin{alignat}{1}
        \bra{\nu'} \Xi &= \bra{\nu'}, \quad
        \braket{\nu'|\tG^{(n)}{}^+|\nu'} \le 1
        \label{eq:nu_}
       \end{alignat}
       holds, then there exists a unitary matrix $D^{(n+1)}$ such that $\bra{\nu} = \bra{\nu'}$.
 \item Equation~\eqref{eq:nu_} holds.
\end{enumerate}

First, we show Property~(a).
Let $\bra{\nu''} \coloneqq \bra{\nu'} \tG^{(n)}{}^\frac{1}{2}{}^+$.
From $\ket{\psi^{(n)}_n} = \Theta \tG^{(n)}{}^\frac{1}{2}\ket{n}$, we obtain
\begin{alignat}{1}
 \braket{\nu''|\Theta^\dagger|\psi^{(n)}_n}
 &= \braket{\nu'|\tG^{(n)}{}^\frac{1}{2}{}^+ \Theta^\dagger \Theta \tG^{(n)}{}^\frac{1}{2}|n}
 = \braket{\nu'|\Xi|n} = tu, \nonumber \\
 \label{eq:nu_Theta_psi}
\end{alignat}
where the second equality follows from
\begin{alignat}{1}
 \tG^{(n)}{}^\frac{1}{2}{}^+ \Theta^\dagger \Theta \tG^{(n)}{}^\frac{1}{2}
 &= \tG^{(n)}{}^\frac{1}{2}{}^+ \Xi \tG^{(n)}{}^\frac{1}{2}
 = \tG^{(n)}{}^\frac{1}{2}{}^+ \tG^{(n)}{}^\frac{1}{2} = \Xi,
\end{alignat}
and the last equality follows from the first equation in Eq.~\eqref{eq:nu_} and
$\braket{\nu'|n} = \braket{n+1|\tG^{(n+1)}|n} = tu$.
Also, from the second equation in Eq.~\eqref{eq:nu_},
we obtain $\braket{\nu''|\Theta^\dagger \Theta|\nu''} = \braket{\nu''|\nu''} \le 1$.
Choose an element $\ket{\check{\nu}'}$ in the null space of $\Theta^\dagger$
such that $\braket{\check{\nu}'|\check{\nu}'} = 1 - \braket{\nu''|\nu''}$, and define
\begin{alignat}{1}
 \ket{\check{\nu}} &\coloneqq \Theta \ket{\nu''} + \ket{\check{\nu}'};
 \label{eq:check_nu}
\end{alignat}
then, $\braket{\check{\nu}|\check{\nu}} = 1$ holds, and thus $\ket{\check{\nu}}$ is normalized.
Multiplying both sides of Eq.~\eqref{eq:check_nu} from the left by $\bra{\psi^{(n)}_n}$
and using
$\braket{\check{\nu}'|\psi^{(n)}_n} = \braket{\check{\nu}'|\Theta \tG^{(n)}{}^\frac{1}{2}|n} = 0$
and Eq.~\eqref{eq:nu_Theta_psi}, we obtain
$\braket{\check{\nu}|\psi^{(n)}_n} = tu = \braket{+|\tU|+}$.
Therefore,
\begin{alignat}{1}
 D^{(n+1)} \ket{\psi^{(n)}_n} &= \ket{+} \ot \ket{0}, \quad
 D^{(n+1)} \ket{\check{\nu}} = \tU^\dagger \ket{+} \ot \ket{0}
 \label{eq:Dnp_psin_checknu}
\end{alignat}
can be satisfied by choosing a unitary matrix of order $dd'$ for $D^{(n+1)}$,
where $\ket{0}$ is a normalized vector in $\V'_n$.
Let
\begin{alignat}{1}
 \tD &\coloneqq D^{(n+1)}{}^\dagger (\tU \ot \I_{d'}) D^{(n+1)}.
\end{alignat}
Substituting Eqs.~\eqref{eq:psik} and \eqref{eq:psin} into Eq.~\eqref{eq:nu}, we obtain
\begin{alignat}{1}
 \bra{\nu} &= \bra{\psi^{(n)}_n} \tD
 \begin{bmatrix}
  \ket{\psi^{(n)}_0} & \ket{\psi^{(n)}_1} & \cdots & \ket{\psi^{(n)}_n} \\
 \end{bmatrix}
 = \bra{\psi^{(n)}_n} \tD \Theta \tG^{(n)}{}^\frac{1}{2}, \nonumber \\
 \label{eq:nu_psi_tD}
\end{alignat}
where the second equality follows from Eq.~\eqref{eq:Theta}.
Also, we have
\begin{alignat}{1}
 \tD \ket{\check{\nu}}
 &= D^{(n+1)}{}^\dagger (\tU \ot \I_{d'}) D^{(n+1)} \ket{\check{\nu}} \nonumber \\
 &= D^{(n+1)}{}^\dagger (\tU \ot \I_{d'}) (\tU^\dagger \ket{+} \ot \ket{0}) \nonumber \\
 &= D^{(n+1)}{}^\dagger (\ket{+} \ot \ket{0}) \nonumber \\
 &= \ket{\psi^{(n)}_n},
\end{alignat}
i.e., $\bra{\psi^{(n)}_n} \tD = \bra{\check{\nu}}$,
where the second and last line follow from Eq.~\eqref{eq:Dnp_psin_checknu}.
Multiplying both sides from the right by $\Theta$ yields
\begin{alignat}{1}
 \bra{\psi^{(n)}_n} \tD \Theta &= \bra{\check{\nu}} \Theta
 = \bra{\nu''} \Xi = \bra{\nu''}.
\end{alignat}
Thus, we have, from Eq.~\eqref{eq:nu_psi_tD},
\begin{alignat}{1}
 \bra{\nu} \tG^{(n)}{}^{\frac{1}{2}+} &= \bra{\psi^{(n)}_n} \tD \Theta
 = \bra{\nu''} = \bra{\nu'} \tG^{(n)}{}^{\frac{1}{2}+}.
\end{alignat}
Multiplying the first and last terms of this equation from the right by $\tG^{(n)}{}^{\frac{1}{2}}$
and using $\bra{\nu} \Xi = \bra{\nu}$
[which follows from Eq.~\eqref{eq:nu_psi_tD}]
and $\bra{\nu'} \Xi = \bra{\nu'}$ of Eq.~\eqref{eq:nu_},
we obtain $\bra{\nu} = \bra{\nu'}$.

Next, we show Property~(b).
As preparation, we compute $\tG^{(n)}{}^+$.
Through algebraic manipulation, we find that $\tG^{(n)}{}^+$ can be expressed in the form
\begin{alignat}{1}
 \braket{i|\tG^{(n)}{}^+|j} &=
 \begin{dcases}
  a_0, & i = j ~\text{and}~ i \equiv_2 0, \\
  a_1, & i = j ~\text{and}~ i \equiv_2 1, \\
  a_2 u^{i-j}, & i \neq j ~\text{and}~ i \equiv_2 j \equiv_2 0, \\
  a_3 u^{i-j}, & i \neq j ~\text{and}~ i \equiv_2 j \equiv_2 1, \\
  a_4 u^{i-j}, & \text{otherwise}. \\
 \end{dcases}
\end{alignat}
For example, we have
\begin{alignat}{1}
 \tG^{(4)}{}^+ =
 \begin{bmatrix}
  a_0 & a_4 u^{-1} & a_2 u^{-2} & a_4 u^{-3} & a_2 u^{-4} \\
  a_4 u & a_1 & a_4 u^{-1} & a_3 u^{-2} & a_4 u^{-3} \\
  a_2 u^2 & a_4 u & a_0 & a_4 u^{-1} & a_2 u^{-2} \\
  a_4 u^3 & a_3 u^2 & a_4 u & a_1 & a_4 u^{-1} \\
  a_2 u^4 & a_4 u^3 & a_2 u^2 & a_4 u & a_0 \\
 \end{bmatrix}.
\end{alignat}
Let
\begin{alignat}{1}
 m &\coloneqq
 \begin{dcases}
  n / 2, & n \equiv_2 0, \\
  (n + 1) / 2, & n \equiv_2 1, \\
 \end{dcases} \nonumber \\
 f_1(l) &\coloneqq m t (1-c^2) + (1 + ltc)(c-t), \nonumber \\
 f_2(l) &\coloneqq (n - m + 1) t (1-c^{-2}) + (1 + ltc^{-1})(c^{-1}-t);
\end{alignat}
then, when $c \neq t$, we have
\begin{alignat}{2}
 a_0 &= \frac{f_1(n-1)}{(1-tc)f_1(n)}, &\quad
 a_1 &= \frac{f_2(n-1)}{(1-tc^{-1}) f_2(n)}, \nonumber \\
 a_2 &= - \frac{tc(c-t)}{(1-tc) f_1(n)}, &\quad
 a_3 &= - \frac{tc^{-1} (c^{-1} - t)}{(1-tc^{-1}) f_2(n)}, \nonumber \\
 a_4 &= - \frac{tc}{f_1(n)}.
\end{alignat}
In this case, since $\tG^{(n)}$ is nonsingular, we have $\Xi = \I_{n+1}$.
When $c = t$, we have
\begin{alignat}{4}
 a_0 &= \frac{f_1(n-1)}{(1-tc)f_1(n)}, &\quad
 a_1 = a_3 &= \frac{1 + (n - m) t^2}{m^2 (1-t^2)}, \nonumber \\
 a_2 &= 0, &\quad
 a_4 &= - \frac{tc}{f_1(n)}
\end{alignat}
and
\begin{alignat}{1}
 \Xi &= \sum_{i=0}^{n - m} \ket{2i}\bra{2i}
 + \frac{1}{m} \sum_{i=1}^{m} \sum_{j=1}^{m} u^{2(i-j)} \ket{2i-1}\bra{2j-1}.
\end{alignat}
In this case, $\tG^{(n)}$ with $n \ge 3$ is not nonsingular.

We first show $\bra{\nu'} \Xi = \bra{\nu'}$.
From Eq.~\eqref{eq:nu}, we have
\begin{alignat}{1}
 \bra{\nu'} &= \sum_{j=0}^n \braket{n+1|\tG^{(n+1)}|j}\bra{j}
 = t \sum_{j=0}^n u^{n+1-j} \sqrt{c_{n+1} c_j} \bra{j},
\end{alignat}
where the second equality follows from Eq.~\eqref{eq:iGnj}.
Thus, we obtain $\bra{\nu'} \Xi = \bra{\nu'}$.

Next, we show $\braket{\nu'|\tG^{(n)}{}^+|\nu'} \le 1$.
Let
\begin{alignat}{1}
 \nu'_k &\coloneqq \braket{\nu'|k} u^{-(n+1-k)} = t \sqrt{c_{n+1} c_k};
\end{alignat}
then, $t \ge \nu'_0 = \nu'_1 c = \nu'_2 = \nu'_3 c = \cdots$ holds.
Considering the case $c \neq t$, we obtain
\begin{alignat}{1}
 \lefteqn{ \braket{\nu'|\tG^{(n)}{}^+|\nu'} } \nonumber \\
 &= \nu'_0{}^2(n - m + 1)[a_0 + (n - m) a_2] \nonumber \\
 &\quad + 2 \nu'_0 \nu'_1 m (n - m + 1) a_4 + \nu'_1{}^2 m [a_1 + (m - 1) a_3]
 \nonumber \\
 &= \frac{\nu'_1{}^2 c^2}{f_1(n)} \left[
 (n - m + 1) \left[ \frac{f_1(m-1)}{1 - tc} - 2tn \right]
 + \frac{m f_2(n - m)}{1 - tc^{-1}} \right] \nonumber \\
 &\le \frac{t^2}{f_1(n)} \left[
 (n - m + 1) \left[ \frac{f_1(m-1)}{1 - tc} - 2tn \right]
 + \frac{m f_2(n - m)}{1 - tc^{-1}} \right] \eqqcolon \tp.
\end{alignat}
Since $f_1(n) - c^2 f_2(n-m) = m t (1-tc)$ holds, we have
\begin{alignat}{1}
 \tp \le 1 &\quad\Leftrightarrow\quad
 (n - m + 1) \left[ \frac{f_1(m - 1)}{1 - tc} - 2tm \right] \nonumber \\
 &\quad\phantom{\Leftrightarrow}\quad
 + \frac{m}{c(c-t)} [f_1(n) - m t(1 - tc)] \le \frac{f_1(n)}{t^2} \nonumber \\
 &\quad\Leftrightarrow\quad
 ntc(c-t) + m t(1-c^2) + c(1-t^2) \ge 0.
\end{alignat}
From $0 \le t \le c \le 1$, the last inequality holds,
which yields $\braket{\nu'|\tG^{(n)}{}^+|\nu'} \le \tp \le 1$.
Considering the case $c = t$, we obtain
\begin{alignat}{1}
 \lefteqn{ \braket{\nu'|\tG^{(n)}{}^+|\nu'} } \nonumber \\
 &= \nu'_0{}^2(n - m + 1)a_0
 + 2 \nu'_0 \nu'_1 m (n - m + 1) a_4 + \nu'_1{}^2 m^2 a_1 \nonumber \\
 &= \nu'_1{}^2 \left[ \frac{(n - m + 1)t^2}{1-t^2}
 + \frac{1 - (n - m + 2)t^2}{1-t^2} \right] \nonumber \\
 &= \nu'_1{}^2 \le 1,
\end{alignat}
where the third line follows from $f_1(n-1) = f_1(n)$
and the inequality follows from $0 \le \nu'_1 \le 1$.

\section{Local unambiguous identification via Bayesian updating} \label{sec:Bayes}

Let the change point be denoted by $k$, and consider the local unambiguous identification
described in Sec.~\ref{sec:LowerBound}.
Let the outcome of the measurement $\{ \Pi^{(n)}_m \}_m$ at step $n \in \{1,\dots,N\}$
be denoted by $s_n \in \{0,1,\question\}$, and for convenience, set $s_0 = 0$.
We assume that if $s_{n-1} = \question$ holds, then the outcome of the measurement
$\{ \Pi^{(n)}_m \}_m$ is either $0$ or $\question$,
i.e., the probability of $s_n = 1$ is zero.
This assumption does not affect the optimality of the average success probability.
For each $n \in \{0,\dots,N\}$, let $S_{n|k}$ and $T_{n|k}$ denote
the conditional probabilities that $s_n$ is $0$ and $\question$, respectively,
given that the change point is $k$.
In the local unambiguous identification described in Sec.~\ref{sec:LowerBound},
\begin{alignat}{1}
 S_{n|n} &= S_{n|n+1} = \cdots = S_{n|N} \eqqcolon S_n, \nonumber \\
 T_{n|n} &= T_{n|n+1} = \cdots = T_{n|N} \eqqcolon T_n
\end{alignat}
holds.
In particular, from $s_0 = 0$, $S_0 = 1$ and $T_0 = 0$ hold.
For each $n \in \{1,\dots,N\}$, let $p_{n|k}$ denote the conditional probability of $s_n = 0$
given that the change point is $k$ and $s_{n-1} = 0$ holds.
Also, for each $n \in \{2,\dots,N\}$, let $q_{n|k}$ denote the conditional probability of
$s_n = 0$ given that the change point is $k$ and $s_{n-1} = \question$ holds.
Similar to $S_{n|k}$ and $T_{n|k}$, we have
\begin{alignat}{1}
 p_{n|n} &= p_{n|n+1} = \cdots = p_{n|N} \eqqcolon p_n, \nonumber \\
 q_{n|n} &= q_{n|n+1} = \cdots = q_{n|N} \eqqcolon q_n.
\end{alignat}
Our goal is to determine the values of $\{ p_n \}_{n=1}^N$ and $\{ q_n \}_{n=2}^N$
that maximize the average success probability.
Since we want to maximize the success probability for $k < N$,
it suffices to consider the case in which the conditional probability of $s_{k+1} = 1$
given that the change point is $k$ and $s_k = 0$ holds is $f(p_k)$.
Let $P(m|k)$ denote the conditional probability that the identification result is $m$
given that the change point is $k$.
Then, when $k < N$, $P(k|k)$ is equal to the joint probability of $s_k = 0$ and $s_{k+1} = 1$,
i.e., $S_k f(p_{k+1})$.
Also, when $k = N$, $P(k|k)$ is equal to the probability of $s_N = 0$, i.e., $S_N$.
Thus, we have
\begin{alignat}{1}
 P(k|k) &=
 \begin{dcases}
  S_k f(p_{k+1}), & k < N, \\
  S_N, & k = N.
 \end{dcases}
\end{alignat}
Figure~\ref{fig:bayes} illustrates the transition of identification outcomes
for the case $k = 2$ and $N \ge 3$.
The average success probability is given by
\begin{alignat}{1}
 \frac{1}{N+1} \sum_{k=0}^N P(k|k)
 &= \frac{1}{N+1} \left[ \sum_{k=0}^{N-1} S_k f(p_{k+1}) + S_N \right].
 \label{eq:Ps_bayes}
\end{alignat}
For each $n \in \{1,\dots,N\}$, $S_n$ is given by
\begin{alignat}{1}
 S_n &= S_{n-1} p_n + T_{n-1} q_n.
 \label{eq:Sn}
\end{alignat}
Also, for each $n \in \{1,\dots,k\}$, since $s_n$ never becomes $1$,
$T_n$ is given by
\begin{alignat}{1}
 T_n &= S_{n-1} (1 - p_n) + T_{n-1} (1 - q_n).
 \label{eq:Tn}
\end{alignat}
From these equations, it is clear that both $S_n$ and $T_n$ are determined by
$\{ p_i \}_{i=1}^n$ and $\{ q_i \}_{i=2}^n$.
\begin{figure}[bt]
 \centering
 \includegraphics[scale=1.0]{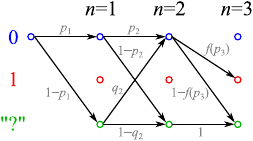}
 \caption{State transition of local identification via Bayesian updating
 for $k = 2$ and $N \ge 3$.
 In this case, the success probability is $S_2 f(p_3)$.}
 \label{fig:bayes}
\end{figure}

First, we consider the optimization of $p_N$ and $q_N$.
Among the terms on the right-hand side of Eq.~\eqref{eq:Ps_bayes},
those that depend on $p_N$ and $q_N$ are $S_{N-1} f(p_N)$ and $S_N$.
From Eqs.~\eqref{eq:Sn} and \eqref{eq:Tn}, we have
\begin{alignat}{1}
 S_{N-1} f(p_N) + S_N &= S_{N-1} f(p_N) + S_{N-1} p_N + T_{N-1} q_N \nonumber \\
 &= S_{N-1} \left[ f(p_N) + p_N \right] + T_{N-1} q_N.
 \label{eq:STN}
\end{alignat}
Therefore, $p_N$ should be chosen to maximize $f(p_N) + p_N$,
and $q_N$ should be chosen as $\ol{p}$.
The optimal values of $p_N$ and $q_N$ do not depend on
$\{ p_i \}_{i=1}^{N-1}$ or $\{ q_i \}_{i=2}^{N-1}$.

Next, for each $n = N-1, N-2, \dots, 1$, we consider optimizing $p_n$ and $q_n$
while fixing $\{ p_i \}_{i=n+1}^N$ and $\{ q_i \}_{i=n+1}^N$.
Assume that there exist constants $A_n$ and $B_n$ such that
\begin{alignat}{1}
 \sum_{k=n}^{N-1} S_k f(p_{k+1}) + S_N &= S_n A_n + T_n B_n.
 \label{eq:ABn}
\end{alignat}
This assumption holds for $n = N - 1$ due to Eq.~\eqref{eq:STN}
(with $A_{N-1} = f(p_N) + p_N$ and $B_{N-1} = q_N$).
Among the terms on the right-hand side of Eq.~\eqref{eq:Ps_bayes},
those that depend on $p_n$ and $q_n$ are
$S_{n-1} f(p_n), S_n f(p_{n+1}), \dots, S_{N-1} f(p_N)$ and $S_N$.
We have
\begin{alignat}{1}
 \sum_{k=n-1}^{N-1} S_k f(p_{k+1}) + S_N
 &= S_{n-1} f(p_n) + S_n A_n + T_n B_n \nonumber \\
 &= S_{n-1} f(p_n) + (S_{n-1} p_n + T_{n-1} q_n) A_n \nonumber \\
 &\quad + \left[ S_{n-1} (1-p_n) + T_{n-1} (1-q_n) \right] B_n \nonumber \\
 &= S_{n-1} \left[ f(p_n) + p_n A_n + (1-p_n) B_n \right] \nonumber \\
 &\quad + T_{n-1} \left[ q_n A_n + (1-q_n) B_n \right] \nonumber \\
 &= S_{n-1} A_{n-1} + T_{n-1} B_{n-1},
\end{alignat}
where the second equality follows from Eqs.~\eqref{eq:Sn} and \eqref{eq:Tn},
and the last equality defines
\begin{alignat}{1}
 A_{n-1} &\coloneqq f(p_n) + p_n A_n + (1-p_n) B_n, \nonumber \\
 B_{n-1} &\coloneqq q_n A_n + (1-q_n) B_n.
\end{alignat}
From $S_{n-1} \ge 0$ and $T_{n-1} \ge 0$, $p_n$ and $q_n$ should be chosen to
maximize $A_{n-1}$ and $B_{n-1}$, respectively.
This is equivalent to maximizing $(A_n - B_n) p_n + f(p_n)$ and $(A_n - B_n) q_n$,
which can be done easily.
The optimal values of $p_n$ and $q_n$ do not depend on
$\{ p_i \}_{i=1}^{n-1}$ or $\{ q_i \}_{i=2}^{n-1}$.
Furthermore, once $p_n$ and $q_n$ are chosen and fixed at their optimal values,
$A_{n-1}$ and $B_{n-1}$ become constants.
Therefore, by redefining $n-1$ as $n$, the assumption that
there exist constants $A_n$ and $B_n$ satisfying Eq.~\eqref{eq:ABn} continues to hold.
Iteratively applying this procedure from $n = N$ to $n = 1$ yields
the optimal values of ${ p_i }{i=1}^N$ and ${ q_i }{i=2}^N$,
with the resulting average success probability given by $A_0 / (N + 1)$.

\section{The case of $\ket{\psi^{(N)}_0}, \dots, \ket{\psi^{(N)}_N}$ being linearly dependent}
\label{sec:Gdep}

In the proof of Theorem~\ref{thm:unamb_opt}, we considered the case
in which $\ket{\psi^{(N)}_0}, \dots, \ket{\psi^{(N)}_N}$ are linearly independent.
We here consider the case in which $\ket{\psi^{(N)}_0}, \dots, \ket{\psi^{(N)}_N}$
are linearly dependent.
Let $\cI$ denote the set of all $i \in \{0,\dots,N\}$ for which there exists
a state $\ket{z}$ satisfying $\braket{z|\psi^{(N)}_j} = \delta_{i,j}$
for any $j \in \{0,\dots,N\}$.
Let $x_i$ denote the success probability when the change point is $i$, for
an arbitrary unambiguous measurement for $\{ \ket{\psi^{(N)}_k} \}_{k=0}^N$.
Then, for each $i \not\in \cI$, $x_i = 0$ must hold; indeed,
no unambiguous measurement exists for which $x_i > 0$ holds.
Let $\bV_{\cI}$ denote the complex vector space spanned by
$\{ \ket{\psi^{(N)}_i} \}_{i \in \cI}$; then,
since $\{ \ket{\psi^{(N)}_i} \}_{i \in \cI}$ is a linearly independent set of vectors,
the dimension of $\bV_{\cI}$ is $|\cI|$, where $|\cI|$ is the number of elements in $\cI$.
For each $i \in \cI$, choose $\ket{\tpsi_i} \in \bV_{\cI}$ such that
$\braket{\tpsi_i|\psi^{(N)}_j} = \delta_{i,j}$ holds for any $j \in \{0,\dots,N\}$,
which is readily seen to be possible.
For each $i \not\in \cI$, define $\ket{\tpsi_i}$ to be the zero vector.
Equation~\eqref{eq:I_nas} is a necessary and sufficient condition
for the existence of an unambiguous measurement for which, for each $i \in \{0,\dots,N\}$,
the success probability is $x_i$ when the change point is $i$.
In what follows, we show that if Eq.~\eqref{eq:G_nas0} holds, then Eq.~\eqref{eq:I_nas} also holds.

Let $\bW_\cI$ denote the complex vector space spanned by $\{ \ket{i} \}_{i \in \cI}$;
then, the dimension of $\bW_\cI$ is $|\cI|$.
For each complex vector space $\bV$, let $P_\bV$ denote the orthogonal projection onto $\bV$.
Then, we have $P_{\bW_\cI} = \sum_{i \in \cI} \ket{i} \bra{i}$.
Let
\begin{alignat}{1}
 \Psi &\coloneqq \sum_{i=0}^N \ket{\psi^{(N)}_i} \bra{i}, \quad
 \Psi_\cI \coloneqq \sum_{i \in \cI} \ket{\psi^{(N)}_i} \bra{i}, \quad
 \Phi \coloneqq \sum_{i=0}^N \ket{\tpsi_i} \bra{i}.
\end{alignat}
From $G = \Psi^\dagger \Psi$ and $\Psi P_{\bW_\cI} = \Psi_\cI$,
multiplying both sides of Eq.~\eqref{eq:G_nas0} from the left and right by $P_{\bW_\cI}$ yields
\begin{alignat}{1}
 \Psi_\cI^\dagger \Psi_\cI &\ge \sum_{i \in \cI} x_i \ket{i} \bra{i}
 = \sum_{i=0}^N x_i \ket{i} \bra{i},
 \label{eq:G_nas_cI}
\end{alignat}
where the equality follows from $x_i = 0$ for any $i \not\in \cI$.
$\Psi_\cI$ and $\Phi$ can be regarded as linear maps from $\bW_\cI$ to $\bV_\cI$.
From
\begin{alignat}{1}
 \Phi^\dagger \Psi_\cI &= \sum_{j=0}^N \ket{j} \sum_{i \in \cI} \braket{\tpsi_j|\psi^{(N)}_i} \bra{i}
 = \sum_{i \in \cI} \ket{i} \bra{i} = P_{\bW_\cI},
\end{alignat}
$\Psi_\cI \colon \bW_\cI \to \bV_\cI$ is bijective, and
$\Phi^\dagger \colon \bV_\cI \to \bW_\cI$ is its inverse.
Multiplying the first and last terms of Eq.~\eqref{eq:G_nas_cI} from the left and right by $\Phi$
and $\Phi^\dagger$, respectively, and using $\Psi_\cI \Phi^\dagger = P_{\bV_\cI}$, we have
\begin{alignat}{1}
 P_{\bV_\cI} &\ge \sum_{i=0}^N x_i \ket{\tpsi_i} \bra{\tpsi_i}.
\end{alignat}
Thus, from $\I \ge P_{\bV_\cI}$, we obtain Eq.~\eqref{eq:I_nas}.
Note that, conversely, if Eq.~\eqref{eq:I_nas} holds, then multiplying both sides
from the left and right by $\Psi^\dagger$ and $\Psi$, respectively, yields Eq.~\eqref{eq:G_nas0}.

\bibliographystyle{apsrev4-1}

% \bibliography{quant}

\end{document}